\documentclass[a4paper]{article}      

\usepackage{amsmath,amssymb,amsthm}
\usepackage{enumerate}
\usepackage{times}
\usepackage[tmargin=20mm,bmargin=20mm,lmargin=25mm,rmargin=25mm,includeheadfoot]{geometry}

\newcommand{\integs}{\mathbb{Z}}
\newcommand{\R}{\mathbb{R}}
\newcommand{\C}{\mathbb{C}}
\newcommand{\real} {\R\mathrm{e}}

\newcommand{\id}{\mathrm{I}}
\newcommand{\srad}{\boldsymbol{\rho}}
\newcommand{\sw}{\sigma}
\newcommand{\s}{N}
\newcommand{\Ind}{\underline{\s}}
\newcommand{\Met}{\mathcal{M}}
\newcommand{\dfn}{\triangleq}
\newcommand{\rb}{\mathbf{z}}
\newcommand{\bw}{\mathbf{w}}
\newcommand{\cle}{\preceq}
\newcommand{\cge}{\succeq}

\newcommand{\diag}{\mathrm{diag}}
\newcommand{\ab}{\mathbf{a}}

\newtheorem{rem}{Remark}
\newtheorem{theorem}{Theorem}
\newtheorem{cor}[theorem]{Corollary}
\newtheorem{defn}{Definition}
\newtheorem{lem}{Lemma}

\begin{document}


  \title{Bounds and Invariant Sets for a Class of Switching Systems
    with Delayed-state-dependent Perturbations}


  \author{Hernan Haimovich$^\dagger$\quad and\quad Mar\'{\i}a M. Seron$^\ddagger$\\[1em]
    $^\dagger$ CONICET and Departamento de Control, Escuela de
    Ing. Electr\'onica,\\
    Universidad Nacional de Rosario, Riobamba 245bis, 2000 Rosario, Argentina.\\[1mm]
    $^\ddagger$ Centre for Complex Dynamic Systems and Control,\\
    The University of Newcastle, Callaghan, NSW 2308, Australia.}

  \date{}

\maketitle

\begin{abstract}
  We present a novel method to compute componentwise transient bounds,
  componentwise ultimate bounds, and invariant regions for a class of
  switching continuous-time linear systems with perturbation bounds
  that may depend nonlinearly on a delayed state. The main advantage
  of the method is its componentwise nature, i.e. the fact that it
  allows each component of the perturbation vector to have an
  independent bound and that the bounds and sets obtained are also
  given componentwise. This componentwise method does not employ a
  norm for bounding either the perturbation or state vectors, avoids
  the need for scaling the different state vector components in order
  to obtain useful results, and may also reduce conservativeness in
  some cases.
  We give conditions for the derived bounds to be of local or
  semi-global nature. In addition, we deal with the case of
  perturbation bounds whose dependence on a delayed state is of affine
  form as a particular case of nonlinear dependence for which the
  bounds derived are shown to be globally valid. A sufficient
  condition for practical stability is also provided.  
  The present paper builds upon and extends to switching systems with
  delayed-state-dependent perturbations previous results by the
  authors. In this sense, the contribution is three-fold: the
  derivation of the aforementioned extension; the elucidation of the
  precise relationship between the class of switching linear systems
  to which the proposed method can be applied and those that admit a
  common quadratic Lyapunov function (a question that was left open in
  our previous work); and the derivation of a technique to compute a
  common quadratic Lyapunov function for switching linear systems with
  perturbations bounded componentwise by affine functions of the
  absolute value of the state vector components. In this latter case,
  we also show how our componentwise method can be combined with
  standard techniques in order to derive bounds possibly tighter than those
  corresponding to either method applied individually.
\end{abstract}



\section{Introduction}
\label{sec:introduction}

Switched systems are dynamical systems that combine a finite number of
subsystems by means of a switching rule
\cite{linant_tac09,liberzon_book03}. The stability of switched systems
has attracted considerable research attention in recent years
\cite{libmor_csm99,decbra_pieee00,showir_siamrev07,linant_tac09}. In
this paper we are concerned with stability under ``arbitrary
switching'', which refers to problems where the stability properties
of interest hold for every admissible switching signal. In this
context, we refer to a switched system undergoing arbitrary switching
as a \emph{switching system}, and as a \emph{switching linear system}
if the individual subsystems have linear dynamics. Necessary and
sufficient conditions for the asymptotic stability of the zero
solution of a switching linear system were given in \cite[Theorem
3]{MoP89} and \cite[Theorem~4.1 and Remark~4.1]{Bla00}. In the present
paper we will focus on the ``practical stability'' problem of
analysing the existence and computation of invariant sets and ultimate
bounds for the switching system state trajectories. This type of
stability is important in every practical setting where nonvanishing
perturbations (also named persistent disturbances) may act on the
system \cite[Ch.~9]{khalil_book02}. We consider switching systems with
a switching linear nominal (unperturbed) system affected by
perturbations that may be nonvanishing and depend nonlinearly on a
delayed state.

Standard methods for the computation of bounds and invariant sets are
based on the use of a Lyapunov function
\cite{khalil_book02}. Arguably, Lyapunov-function-based methods are
the most powerful and widely applicable, although their inherent
difficulty is the obtention of a suitable Lyapunov function. When the
nominal system is linear, however, a quadratic Lyapunov function can
easily be computed via solving a Lyapunov equation, but the bounds so
obtained may be conservative, even for linear systems (see, e.g.,
Section 1 of \cite{kofhai_ijc07}). State bounds computed by means of a
quadratic Lyapunov function are given as a bound on the norm, usually
the 2-norm, of the state vector and usually require a bound on the
norm of the perturbation vector.  The aforementioned conservativeness
may be due to (a) the information on the different bounds for each
component of the perturbation vector is lost when taking its norm and
(b) the bounds corresponding to different state vector components are
substantially different and hence its 2-norm is not the most suitable
for bounding. Problem (b) may be ameliorated by properly scaling the
state vector components. In order to avoid or at least reduce the
effect of both problems (a) and (b), then Lyapunov functions of a form
more complicated than quadratic may be employed. Likewise, for
switching systems with a switching linear nominal system, a quadratic
Lyapunov function common to all linear subsystems can be computed via
linear matrix inequalities (LMIs) in case one exists (see, for
example, Section~4.3 of \cite{showir_siamrev07} and the references
therein).  As in the non-switching case, the bounds thus obtained may
be conservative in some
cases.  


The present paper follows a methodology which differs from the one
just described in that the use of either a norm of the state or a
Lyapunov function can be avoided. Moreover, this methodology can be
easily combined with Lyapunov analysis in order to possibly improve on
the results of either method applied individually. The methodology
that we employ is based on componentwise analysis, avoids the need for
scaling individual state components, and builds upon and extends to
switching systems with delayed-state-dependent perturbations previous
results of
\cite{kofhai_ijc07,kofserhai_auto08,haiser_cdc09,haiser_auto10}. In
\cite{kofhai_ijc07}, a method to compute componentwise ultimate bounds
for perturbed (non-switching) linear systems is given. The
perturbation bound is allowed to depend nonlinearly on the system
state. Ultimate bounds are derived that are global (valid for every
initial condition) when the perturbation bound is constant and local
(valid only when the initial state is in a specific region) in the
more general case of state-dependent perturbation bounds. Global
componentwise ultimate bounds for perturbation bounds that have affine
dependence on a delayed system state are derived in Section~3 of
\cite{kofserhai_auto08}, jointly with a sufficient condition for
practical stability. In \cite{haiser_cdc09,haiser_auto10}, a method to
derive global componentwise transient and ultimate bounds was proposed
for a class of switching linear systems with constant perturbation
bounds. It was shown in \cite{haiser_auto10} that the proposed method
can be applied when the switching linear system is close to being
simultaneously triangularizable. In such a case, a common quadratic
Lyapunov function (CQLF) exists for the switching system. However, the
precise relationship between the class of switching linear systems to
which the proposed method can be applied and those that admit a CQLF
was left as an open question.

The present paper provides three contributions. The first contribution
is to answer the aforementioned open question: the class of switching
linear systems to which our componentwise bound and invariant set
method can be applied is strictly contained in the class of switching
linear systems that admit a CQLF, although the switching linear system
need not be close to simultaneously triangularizable. This
relationship was reported by Mori et al. in \cite{mormor_cdc01} but
the proof was not given. We provide a proof and, moreover, extend it
so that it becomes useful in the derivation of our third
contribution. The second contribution of the paper is to combine and
extend the previous results in
\cite{kofhai_ijc07,kofserhai_auto08,haiser_cdc09,haiser_auto10} by
providing transient bounds, ultimate bounds, and invariant regions
based on componentwise analysis for a class of switching
continuous-time linear systems with perturbation bounds that may
depend nonlinearly on a delayed state. This kind of setting can
describe, for example, switching linear systems with uncertainty in
the state evolution matrix, switching linear systems with an uncertain
time delay and, more generally, switching nonlinear systems expressed
as their switching linear approximation perturbed by an additive
disturbance with a bound depending nonlinearly on the system state. We
derive conditions for the bounds to be of local or semi-global
nature. We also address the particular case of perturbation bounds
that have affine dependence on a delayed state. In this particular
case, the bounds derived are shown to be of global nature and an
extension of {the sufficient condition} for practical stability of
\cite[Section~3]{kofserhai_auto08} is provided. The third contribution
is to provide a technique to compute a CQLF for a class of switching
linear systems with perturbations bounded componentwise by affine
functions of the absolute value of the state vector components
(provided no delays are present). The CQLF so derived can be used to
compute ultimate bounds for this class of systems. Moreover, both the
componentwise method and the Lyapunov technique can be combined to
obtain tighter bounds than could be obtained by either methodology
applied individually. The combination of both methodologies is
illustrated by means of a numerical example. The current paper
subsumes all the aforementioned previous bound computation results
\cite{kofhai_ijc07,kofserhai_auto08,haiser_auto10} for (switching and
non-switching) continuous-time systems, in the sense that bounds for
each of the cases considered in these results can be obtained by means
of the current results (although the bounds obtained may not be
identical). Although similar ideas are employed, the extension of the
previous results to derive the ones presented in the current paper is
not straightforward.
%
Some of the results in the current paper have been presented in
\cite{haiser_aucc11,haiser_rpic11}.

The remainder of the paper proceeds as follows. We conclude this
introductory section with a summary of the notation employed
throughout the paper. Section~\ref{sec:problem-formulation} presents
the problem formulation together with some preliminary definitions and
properties. Section~\ref{sec:cont-time-switch} contains the main
results of the paper, and is organised into four subsections
presenting, respectively, an overview of previous results for constant
perturbation bounds, the connection between the latter results and the
existence of a CQLF, the new results for the case of nonlinear
perturbation bounds, and the new results for the special case of
affine perturbation bounds, including the connection with CQLF when no
delay is present. Section~\ref{sec:ct-example} illustrates the results
by means of a numerical example. Section~\ref{sec:conc} provides
conclusions and outlines directions for future work. To ease
readability, proofs are provided in the appendix.

\textbf{Notation.}  $\integs$, $\R$ and $\C$ denote the sets of
integer, real and complex numbers, and $0$ denotes the zero scalar,
vector or matrix, depending on the context. $\R_{+}$ and $\R_{+0}$
denote the positive and nonnegative real numbers, respectively, and
similarly for $\integs_{+}$ and $\integs_{+0}$. If $M$ is a matrix,
then $M'$ denotes its transpose, $M^*$ its conjugate transpose, and
$|M|$ is the matrix whose entries are the magnitude of the
corresponding entries in $M$. If $P$ is a square matrix, then
$\srad(P)$ denotes its spectral radius, $\ab(P)$ its spectral abscissa,
and $P>0$ ($P<0$) means that $P$ is positive (negative) definite. If
$x(t)$ is a vector-valued function, then $\limsup_{t \rightarrow
  \infty} x(t)$ denotes the vector obtained by taking $\limsup_{t
  \rightarrow \infty}$ of each component of $x(t)$. Similarly,
`$\lim$' and `$\max$' denote componentwise operations on a vector or
matrix. The expression $x \preceq y$ ($x \prec y$) denotes the set of
componentwise inequalities $x_i \le y_i$ ($x_i < y_i$) between the
elements of the real vectors $x$ and $y$, and similarly for $x \succeq
y$ ($x \succ y$) and in the case when $x$ and $y$ are matrices. If $T
: \R^n_{+0} \to \R^n_{+0}$, then $T^k$ denotes the iteration of $T$,
that is, the maps defined by $T^1(x) = T(x)$ and $T^{k+1}(x) =
T(T^k(x))$. The index set $\{1,2,\ldots,N\}$ is denoted $\Ind$ and
$\mathbf{i}$ denotes $\sqrt{-1}$. Employing this notation, note that
$P\succ 0$ means that every entry of $P$ is positive and $P>0$ that
$P$ is positive definite.


\section{Problem Formulation}
\label{sec:problem-formulation}

In this section, we formulate the problem to be addressed, followed by
some preliminary definitions and properties.


\subsection{Problem statement}
\label{sec:problem-statement}

We consider switching continuous-time 
perturbed systems of the form
\begin{align}
  \label{eq:system}
  \dot{x}(t) &=A_{\sw(t)} x(t) + H_{\sw(t)} w_{\sw(t)}(t),
\end{align}
where $x(t) \in \R^n$ is the system state, $\sw(t) \in \Ind \triangleq
\{1,2, \dots, N\}$ is the switching function, $A_i \in \R^{n\times
  n}$, $H_i \in \R^{n\times k_i}$ for $i\in\Ind$, and the perturbation
vectors $w_i(t) \in \R^{k_i}$ satisfy the componentwise bound
\begin{equation} \label{eq:boundw}
  |w_i(t)|\preceq \delta_i( \theta(t) ) \; \text{ for all } t\geq 0,
  \; \text{ for }i\in\Ind,
\end{equation}
with continuous bounding functions $\delta_i : \R^n_{+0} \to
\R^{k_i}_{+0}$ and $\theta(t) \in \R^n_{+0}$ defined as
\begin{equation} \label{eq:barx}
  \theta(t)\dfn \max_{t-\bar{\tau}\leq \tau \leq t} |x(\tau)|,
\end{equation}
where $\bar{\tau}\geq 0$ and the maximum is taken componentwise. 

Note that for each $i\in\Ind$, \eqref{eq:boundw} expresses a bound for
each one of the $k_i$ components of the perturbation vector $w_i(t)$,
and that the maximum in (\ref{eq:barx}) denotes a componentwise
operation.
\begin{rem} 
  The setting \eqref{eq:system}--\eqref{eq:barx} can
  describe, inter-alia, the following situations:
  \begin{itemize}
  \item Uncertainty in the system evolution matrix, where $
    \dot{x}(t)$ 
    has the form
    $(A_{\sw(t)} + \Delta A_{\sw(t)}(t))x(t)$, and $|\Delta A_i(t)|
    \preceq \overline{\Delta A_i}$, for all $t \ge 0$ and $i\in\Ind$;
    in this case, we can take $H_i=\id$ in \eqref{eq:system}
    , $\delta_i(\theta) =\overline{\Delta A
    }\theta$ in~\eqref{eq:boundw}, and $\bar{\tau} = 0$ in
    \eqref{eq:barx}.
  \item Uncertain time delays, where $ w_i(t) = F_i x(t-\tau_i)$, and
    $0\le \tau_i \le \tau_{\max}$; in this case, we can take
    $\delta_i(\theta) = |F_i|\theta$ in~\eqref{eq:boundw}, and
    $\bar{\tau} = \tau_{\max}$ in \eqref{eq:barx}.
  \item Disturbances with constant bounds: $\delta_i(\theta)=\bw_i$
    in~\eqref{eq:boundw}.
  \item Switching nonlinear systems where $ \dot{x}(t)$ 
    has the form 
    $f_{\sw(t)}(x(t))$; in this case we may take $A_i = \frac{\partial
      f_i}{\partial x}(x_0)$, $H_i = \id$, $\bar\tau=0$,
    $\delta_i(\theta) = \max_{x : |x| \cle \theta} |f_i(x) - A_ix|$.
  \end{itemize}
\end{rem}

The problem of interest is to derive transient bounds, ultimate
bounds, and invariant sets for switching systems of the form
(\ref{eq:system}) 
with
perturbations bounded as in (\ref{eq:boundw})--(\ref{eq:barx}). This
will be addressed in Section~\ref{sec:cont-time-switch}
. In the next
subsection, we give some definitions and preliminary results related
to the concept of Metzler matrices and to a specific class of
nonnegative functions.

\subsection{Definitions and properties}
\label{sec:definitions}

\begin{defn}[Metzler]
  \label{defn:Metzler}
  A matrix $\Lambda \in \R^{n\times n}$ is \emph{Metzler} if its
  off-diagonal entries are nonnegative.
\end{defn}

Given an arbitrary matrix $N\in\C^{n\times n}$, we define
$\Met(N)\in\R^{n\times n}$ as the matrix whose entries satisfy
\begin{equation}
  \label{eq:31a}
  [\Met(N)]_{i,k} =
  \begin{cases}
    \real\{N_{i,k}\} &\text{if $i=k$,}\\
    |N_{i,k}|        &\text{if $i\neq k$.}
  \end{cases}
\end{equation}
Note that $\Met(N)$ is Metzler for every $N\in\C^{n\times n}$.

The following Lemma gives properties of Metzler matrices.
\begin{lem}
  \label{lem:propMetzler}
  Let $\Lambda,M \in \R^{n\times n}$ and $N\in\C^{n\times n}$. Then,
  \begin{enumerate}[a)]
  \item $\Lambda$ is Metzler if and only if $e^{\Lambda t} \succeq 0$
    for all~$t\ge 0$.
  \item If $\Lambda$ is Metzler, then it is Hurwitz if and only if $-\Lambda^{-1}
    \succeq 0$.
  \item $\Lambda$ is Metzler and Hurwitz if and only if $-\Lambda$ is
    an M-matrix.
  \item If $M=M'$ and is Metzler, then $x'Mx \le |x|' M |x|$ for all $x\in\R^n$.
  \item If $N=N^*$, then $z^* N z \le |z|' \Met(N) |z|$ for all $z\in\C^n$.
  \end{enumerate}
\end{lem}
Properties a) and b) can be found in Chapter 6 of
\cite{luenbe_book79}; c) follows from Definition~\ref{defn:Metzler}
and the definition of an M-matrix (see, e.g., Chapter 6 of
\cite{berple_book94}); d) and e) are straightforward.

\begin{defn}[CNI]
  \label{defn:CNI}
  A nonnegative vector function $f: \R^n_{+0} \to \R^m_{+0}$ is said
  to be \emph{Componentwise Non-Increasing} (CNI) if, whenever $x_1,x_2
  \in \R^n_{+0}$ and $x_1 \preceq x_2$, then $f(x_1) \preceq f(x_2)$.
\end{defn}

\begin{rem}
  \label{rem:cni}
  Every continuous function $\hat f: \R^n_{+0} \to \R^m_{+0}$ can be
  overbounded by a continuous CNI function. In particular, the
  tightest continuous CNI overbound of $\hat f$ is the function $f:
  \R^n_{+0} \to \R^m_{+0}$ given by
  \begin{equation}
    \label{eq:46}
    f(x) = \max_{0\cle y\cle x} \hat f(y).
  \end{equation}
\end{rem}

\section{Main Results}
\label{sec:cont-time-switch}

In this section, we begin by briefly reviewing in
Section~\ref{sec:prev-res} our previous result
(Theorem~\ref{thm:mainct} below) for switching linear systems with
constant perturbation bounds \cite{haiser_auto10}.
Section~\ref{sec:relationship-cqlf} provides the first contribution of
the paper by establishing the link between the applicability of the
previous results of \cite{haiser_auto10} and that of the CQLF, a
question that was left open in the latter reference.
The main results of the paper are given in
Sections~\ref{sec:nonlin-pert-bnd} and~\ref{sec:lin-pert-bnd}. In
Section~\ref{sec:nonlin-pert-bnd}, we provide novel transient bounds,
ultimate bounds and invariant sets for a class of switching
continuous-time linear systems with perturbations bounded by a
nonlinear function of a delayed state. In
Section~\ref{sec:lin-pert-bnd}, we provide additional results for the
special case of perturbation bounds having affine dependence on a
delayed state and also show how to compute a CQLF when no delay is
present.  The proofs are given in the appendix.

\subsection{Previous results: Constant perturbation bounds}
\label{sec:prev-res}

The following is a minor modification of Theorem~1
of~\cite{haiser_auto10}.
\begin{theorem}[Theorem~1 of \cite{haiser_auto10}]
  \label{thm:mainct}
  Consider the switching system (\ref{eq:system}) with componentwise
  perturbation bound
  \begin{equation}
    \label{eq:7a}
    |w_i(t)| \preceq \bw_i,
  \end{equation}
  with $\bw_i \in \R^{k_i}_{+0}$. Let $V\in\C^{n\times n}$ be
  invertible and define
  \begin{align}
    \label{eq:88}
    \Lambda_i &\dfn V^{-1}A_i V,\quad M_i \dfn \Met(\Lambda_i),\quad
    \Lambda \dfn \max_{i\in\Ind} M_i
  \end{align}
  where $\Met(\cdot)$ is the operation defined in
  \eqref{eq:31a}. Suppose that $\Lambda$ is Hurwitz. Let $\rb \in
  \R^n_{+0}$ satisfy
  \begin{align}
  \label{eq:9a}
    \rb &\succeq \max_{i\in\Ind} \left[\max_{|w_i|\preceq
        \bw_i} |V^{-1}H_i w_i|\right],
  \end{align}
  and define
  \begin{align}
    \label{eq:86}
    \eta &\dfn \max\big\{ |V^{-1}x(0)|+\Lambda^{-1}\rb, 0 \big\}.
  \end{align}
  Then, the states of system~(\ref{eq:system}) are bounded as
  \begin{align}
    \label{eq:80}
    &|V^{-1}x(t)| \preceq -\Lambda^{-1}\rb + e^{\Lambda t} \eta,
  \end{align}
  for all $t\ge 0$, and 
  ultimately bounded as
  \begin{align}
    \label{eq:29a}
    &\limsup_{t\to\infty} |V^{-1}x(t)| \preceq -\Lambda^{-1} \rb.
  \end{align}
\end{theorem}

\begin{rem}
  \label{rem:opt}
  The main assumption that enables the application of
  Theorem~\ref{thm:mainct} is the obtention of an invertible matrix
  $V$ so that $\Lambda$ in (\ref{eq:88}) be Hurwitz. In
  \cite{haiser_auto10}, an algorithm to seek such a matrix was
  provided. This algorithm searches over unitary matrices
  $V$. However, it may happen that even if a matrix $V$ that makes
  $\Lambda$ Hurwitz exists, no unitary matrix $V$ ensuring such a
  condition exists. A general algorithm to seek the required matrix
  $V$ is the following. Let $\ab(\Lambda)$ denote the spectral
  abscissa of $\Lambda$, i.e. the maximum over the real parts of the
  eigenvalues of $\Lambda$. We pose the following optimization
  problem:
  \begin{quote}
    Minimize $\ab(\Lambda)$ over $V\in\C^{n\times n}$ invertible.
  \end{quote}
  It is not necessary to find a global optimum of this nonconvex
  optimization problem: it suffices to find an invertible $V$ such
  that $\ab(\Lambda) < 0$, i.e. such that $\Lambda$ is Hurwitz. Note
  that for every nonzero scalar $\alpha\in\C$, according to
  (\ref{eq:88}) the matrices $V$ and $\alpha V$ will produce the same
  $\Lambda_i$ and hence the same $\Lambda$. Consequently, when
  searching for a suitable $V$ according to the above optimization,
  the entries of $V$ can be bounded a priori without affecting the
  success of the search.
\end{rem}
\begin{rem}
  \label{rem:shape}
  A region of the form $\{x \in \R^n : |V^{-1} x| \preceq \bar z\}$,
  with $\bar z \succeq 0$ as given by (\ref{eq:80}) and
  (\ref{eq:29a}), has polyhedral shape if the entries of $V$ are real,
  and a combined ellipsoidal/polyhedral shape if $V$ has some complex
  entries (see \cite{haikof_iwc08} for more details). Every
  (componentwise) bound $|V^{-1} x| \preceq \bar z$ yields a
  corresponding componentwise bound $|x| \preceq |V| \bar z$, since
  \begin{equation}
    \label{eq:12a}
    |x| = |VV^{-1}x| \preceq |V| |V^{-1}x| \preceq |V| \bar z.
  \end{equation}
\end{rem}

\subsection{Relationship to CQLF}
\label{sec:relationship-cqlf}

The following result establishes the relationship between the
existence of the matrix $V$ required by Theorem~\ref{thm:mainct} and
the existence of a quadratic Lyapunov function. A similar result has
been reported in \cite{mormor_cdc01}, where the class of systems for
which the matrix $V$ required by Theorem~\ref{thm:mainct} exists was
identified as a subclass of the switching systems that admit a
CQLF. The result in \cite{mormor_cdc01} was stated without proof, nor
reference to another publication containing the proof. Here we provide
a proof and, moreover, will present an extension
[Theorem~\ref{thm:UBlinbnd}(\ref{item:11}) in
Section~\ref{sec:lin-pert-bnd}] where sufficient conditions for the
existence of a CQLF guaranteeing practical stability are given for the
case of perturbations bounded by an affine function of the
(non-delayed) state.
\begin{theorem}
  \label{thm:cdlf}
  Let $\Lambda\in\R^{n\times n}$ be Metzler and let $\bar\Lambda$ be
  Hurwitz and satisfy $\bar\Lambda \cge \Lambda$. Then,
  \begin{enumerate}[a)]
  \item there exists a diagonal and positive definite matrix
    $D=\diag(d_1,\ldots,d_n)>0$ satisfying
    \begin{equation}
      \label{eq:2}
      \bar\Lambda'D+D\bar\Lambda < 0;
    \end{equation}\label{item:1}
  \item $\Lambda$ is Hurwitz.\label{item:2}
  \item If $\Lambda$ satisfies (\ref{eq:88}) for some invertible
    $V\in\C^{n\times n}$ and matrices $A_i\in\R^{n\times n}$, then for
    each $D$ as in \ref{item:1}), the corresponding real symmetric and
    positive definite matrix $P=\real\{(V^{-1})^* D V^{-1}\}$
    satisfies
    \begin{equation}
      \label{eq:1}
      A_i'P + P A_i < 0,\qquad \text{for all }i\in\Ind.
    \end{equation}\label{item:3}
  \end{enumerate}
\end{theorem}

The following consequence of Theorem~\ref{thm:cdlf} constitutes an
important fact regarding Metzler and Hurwitz matrices and the
operation (\ref{eq:31a}).
\begin{cor}
  \label{cor:lamHur}
  Consider the switching system (\ref{eq:system}), let
  $V\in\C^{n\times n}$ be invertible, and define $\Lambda_i$ and $\Lambda$
  as in (\ref{eq:88}), where $\Met(\cdot)$ is the operation defined in
  (\ref{eq:31a}). If $\Lambda$ is Hurwitz, then $A_i$ is Hurwitz for
  all $i\in\Ind$. Moreover, the $A_i$ admit a common quadratic
  Lyapunov function.
\end{cor}
\begin{proof}
  Just apply Theorem~\ref{thm:cdlf}\ref{item:3}) with $\bar\Lambda = \Lambda$.
\end{proof}

The above theorem and corollary establish that the class of switching
systems considered in this paper, that is, those for which the matrix
$V$ required by Theorem~\ref{thm:mainct} exists, admit a common
quadratic Lyapunov function. This closes a problem left open in our
previous paper \cite{haiser_auto10}. As shown previously in
\cite{haiser_auto10} and \cite{mormor_cdc01}, the class of switching
systems considered in the present paper \emph{contains} the class of
systems that can be simultaneously triangularized by means of a common
transformation. Moreover, the class of switching systems considered is
not a trivial extension of the class of switching systems admitting
simultaneous triangularization.  To illustrate this point, we revisit
the example presented in \cite{DaM99} consisting of
system~\eqref{eq:system} with no disturbance, $\sw(t) \in \{1,2\}$ and
  \[ A_1=
  \begin{bmatrix}
    -1 &-1 \\1 & -1
  \end{bmatrix} , \quad
A_2=
\begin{bmatrix}
  -1 & -a \\1/a  & -1
\end{bmatrix}.
\] %
Note that for every value of $a$, the eigenvalues of $A_2$ are
$-1\pm\mathbf{i}$, identical to those of $A_1$, and hence both $A_1$
and $A_2$ are Hurwitz. However, the eigenvectors of $A_1$ are $[1,\
\pm\mathbf{i}]'$ and those of $A_2$ are $[1,\ \pm a \mathbf{i}]'$. In
order to be simultaneously triangularizable, it is necessary that both
$A_1$ and $A_2$ have a common eigenvector. Consequently, loosely
speaking we may say that this switching system is farther away from
simultaneous triangularization as $a$ is varied farther away from
1. It was shown in \cite{DaM99} that for $a > 3 +\sqrt{8}$ the above
switching system does not admit a CQLF. For $a=3 +\sqrt{8} - 10^{-3}$,
which corresponds to a switching system with stable subsystems but so
far from simultaneous triangularization that it is at the verge of not
admitting a CQLF, searching for a unitary $V$ by means of the
algorithm in \cite{haiser_auto10} yields a solution for which
$\Lambda$ is not Hurwitz. However, searching for an arbitrary $V$ by
means of the optimization proposed in Remark~\ref{rem:opt}, we are
able to obtain the feasible solution
\[ V=  \left[
    \begin{smallmatrix}
      -6.0069   & 5.5729 \\
   -0.3554   &-1.0843
  \end{smallmatrix}\right] +
    \left[
      \begin{smallmatrix}
        0.8605  & -2.6151\\
        -2.4885  & -2.3081
     \end{smallmatrix}\right]\mathbf{i},
\]
for which the corresponding $\Lambda$ is Hurwitz.

{In addition, the class of switching systems considered in this paper
  is \emph{strictly} contained in the class of switching linear
  systems that admit a CQLF, i.e., some switching systems may admit a
  CQLF but the matrix $V$ required by Theorem~\ref{thm:mainct} may not
  exist. To see this, consider Example 4.1 of \cite{shonar_acc00},
  which consists of system~\eqref{eq:system} with no disturbance,
  $\sw(t) \in \{1,2,3\}$ and
  \[ A_1=
  \begin{bmatrix}
    0 &5\\ -30 & -1.4
  \end{bmatrix},\, 
  A_2 =
  \begin{bmatrix}
    0 & 5\\ -26 & -1
  \end{bmatrix},\, 
  A_3 =
  \begin{bmatrix}
    -6 & 27\\ -150 & -1
  \end{bmatrix} \]
  This switching system admits a CQLF but the search for $V$ outlined
  in Remark~\ref{rem:opt} does not give a useful solution, even when
  the optimization is run over 1000 times from different arbitrary
  initial conditions.
}

\subsection{Nonlinear perturbation bounds}
\label{sec:nonlin-pert-bnd}

Theorem~\ref{thm:nlpertct} below establishes local transient and
ultimate bounds for system (\ref{eq:system}) with perturbation bounds
of the form (\ref{eq:boundw})--(\ref{eq:barx}). The theorem is
followed by the derivation of invariant regions
(Corollary~\ref{cor:nlinv}) and of conditions for the bounds to be of
semi-global nature (Corollary~\ref{cor:gub}).

\begin{theorem}
  \label{thm:nlpertct}
  Consider the switching system (\ref{eq:system}) with perturbation
  bound of the form (\ref{eq:boundw})--(\ref{eq:barx}), where the
  bounding functions $\delta_i$ are CNI. Let $V\in\C^{n\times n}$ be
  invertible and define $\Lambda_i$ for $i\in\Ind$ and $\Lambda$ as in
  (\ref{eq:88}), where $\Met(\cdot)$ is the operation defined in
  (\ref{eq:31a}). Suppose that $\Lambda$ is Hurwitz. Let $\psi :
  \R^n_{+0} \to \R^n_{+0}$ be defined as in (\ref{eq:40}), let $\delta
  : \R^n_{+0} \to \R^n_{+0}$ be continuous, CNI and satisfy
  (\ref{eq:6}), and for every $\gamma \in \R^n_{+0}$ consider
  $T_\gamma : \R^n_{+0} \to \R^n_{+0}$ defined in (\ref{eq:19}).
  \begin{align}
    \label{eq:40}
    \psi(x) &= \max_{i\in\Ind} \left[ \max_{|w_i| \preceq \delta_i( |V| x )} |V^{-1} H_i w_i| \right],\\
    \label{eq:6}
    \delta(x) &\succeq \psi(x),\quad \text{for all }x \in \R^n_{+0},\\
    \label{eq:19}
    T_\gamma(x) &= -\Lambda^{-1} \delta(x) + \gamma.
  \end{align}
  Suppose that there exists $\beta \in \R^n_{+0}$ satisfying
  $T_0(\beta) \prec \beta$. Then,
  \begin{enumerate}[(a)]
  \item For every $k\in\integs_{+}$, $T_0^{k+1}(\beta) \preceq
    T_0^k(\beta)$ and $\lim_{k\to\infty} T_0^k(\beta) = b \succeq 0$.\label{item:5a}
  \item Transient bounds. For every $\gamma \in \R^n_{+0}$ such that
    $-\Lambda^{-1} [\delta(\beta) + \max\{- \Lambda \gamma, 0\}] \prec
    \beta$, it happens that if $|V^{-1} x(t)| \preceq T_\gamma(\beta)$
    for all $-\bar\tau \le t \le 0$, then $|V^{-1} x(t)| \preceq
    \beta$ for all $t\ge -\bar\tau$.\label{item:6}
  \item Selection of $\gamma\in\R^n_{+}$ for transient bounds. For
    every positive vector $c\in\R^n_{+}$, let $p(c)$ denote the vector
    in $\R^n_{+0}$ whose components satisfy
    \begin{equation}
      \label{eq:68}
      [p(c)]_j =
      \begin{cases}
        (-\Lambda c)_j &\text{if }(-\Lambda c)_j > 0,\\
        0 &\text{if }(-\Lambda c)_j \le 0,
      \end{cases}
    \end{equation}
    for $j=1,\ldots,n$. Then, $p(c) \neq 0$ and for every $\epsilon$ satisfying
    $0 < \epsilon < \bar\epsilon$, where
    \begin{equation}
      \label{eq:67}
      \bar\epsilon \dfn \min_{j:[-\Lambda^{-1}p(c)]_j\neq 0}
      \frac{[\beta+\Lambda^{-1}\delta(\beta)]_j}{[-\Lambda^{-1}p(c)]_j}
      > 0,
    \end{equation}
    it happens that $-\Lambda^{-1}[\delta(\beta) + \max\{-\Lambda
    c\epsilon,0\}] \prec \beta$.\label{item:9}
  \item Ultimate bounds. If $|V^{-1} x(t)| \preceq \beta$ for all $t
    \ge -\bar\tau$, then $\limsup_{t\to\infty} |V^{-1} x(t)| \preceq
    b$.\label{item:7}
  \end{enumerate}
\end{theorem}

In addition to the obtention of $V$ such that
$\Lambda$ is Hurwitz, whose computation is explained in
Remark~\ref{rem:opt}, Theorem~\ref{thm:nlpertct} requires a
nonnegative vector $\beta$ satisfying $T_0(\beta) \prec \beta$. If
such a vector exists, then it can be computed by means of
Algorithm~1 and Theorem~3 of \cite{kofhai_ijc07}.

Theorem~\ref{thm:nlpertct}(\ref{item:5a}) establishes a monotonicity
property of the sequence of vectors obtained by iterating the map
$T_0$ on the vector $\beta$. This property is useful to ensure the
existence of the limiting vector $b$, which constitutes the smallest
componentwise ultimate bound that can be obtained for $|V^{-1} x(t)|$
by direct application of this theorem for the given vector $\beta$
[Theorem~\ref{thm:nlpertct}(\ref{item:7})].

Theorem~\ref{thm:nlpertct}(\ref{item:6}) provides bounds for each of the
components of $|V^{-1} x(t)|$ that are valid at every time instant,
provided the initial condition $|V^{-1} x(t)|$, $-\bar\tau \le t \le
0$, is bounded by $T_\gamma(\beta)$. For the bounds provided by
Theorem~\ref{thm:nlpertct}(\ref{item:6}) to be valid, the existence of
$\gamma \in \R^n_{+0}$ so that $-\Lambda^{-1} [\delta(\beta) + \max\{-
\Lambda \gamma, 0\}] \prec \beta$ is required. Note that substituting
$0$ for $\gamma$ into the latter condition, and recalling
(\ref{eq:19}), yields $T_0(\beta) \prec \beta$, which holds by
assumption. Therefore, such condition always holds for $\gamma=0$, and
by continuity, it will also hold for every $\gamma \in \R^n_{+0}$ with
small enough components. The advantage of employing $\gamma$ with
greater components is a larger set of initial conditions for which the
bound given by Theorem~\ref{thm:nlpertct}(\ref{item:6}) is valid. 

Theorem~\ref{thm:nlpertct}(\ref{item:9}) shows how the aforementioned
vector $\gamma$ can be computed so that all of its components are not
only nonnegative but also positive. Specifically,
Theorem~\ref{thm:nlpertct}(\ref{item:9}) establishes that if an
arbitrary positive vector $c$ is selected, $\gamma=\epsilon c$ will
satisfy the requirement in Theorem~\ref{thm:nlpertct}(\ref{item:6})
for every positive scalar $\epsilon$ satisfying
$\epsilon<\bar\epsilon$ with $\bar\epsilon$ as in (\ref{eq:67}). Note
that there is ample leeway in the selection of $\gamma$, since the
vector $c$ is positive but otherwise arbitrary.

Theorem~\ref{thm:nlpertct}(\ref{item:7}) provides componentwise
ultimate bounds whenever the state remains within the bound given by
Theorem~\ref{thm:nlpertct}(\ref{item:6}) at all times. The combination
of parts~(\ref{item:6}) and (\ref{item:7}) of
Theorem~\ref{thm:nlpertct} gives local ultimate bounds, i.e., ultimate
bounds that are guaranteed to hold for initial conditions within a
certain set.

\begin{cor}[Invariance]
  \label{cor:nlinv}
  In addition to the hypotheses of Theorem~\ref{thm:nlpertct}, suppose
  that for every $\epsilon \in \R^n_+$, there exists $\beta_\epsilon$
  such that $b \preceq \beta_\epsilon \preceq b + \epsilon$, and
  $T_0(\beta_\epsilon) \prec \beta_\epsilon$. Then, if $|V^{-1} x(t)|
  \preceq b$ for all $-\bar\tau \le t \le 0$, then $|V^{-1} x(t)|
  \preceq b$ for all $t \ge -\bar\tau$.
\end{cor}

\begin{cor}[Semi-global ultimate bounds]
  \label{cor:gub}
  In addition to the hypotheses of Theorem~\ref{thm:nlpertct}, suppose
  that for every $\xi \in \R^n_{+0}$ there exist $\beta,\gamma \in
  \R^n_{+0}$ satisfying
  \begin{align}
    \label{eq:43}
    &\xi \preceq T_\gamma(\beta),\quad\text{and}\\
    \label{eq:1a}
    &-\Lambda^{-1} [\delta(\beta) + \max\{- \Lambda \gamma, 0\}] \prec \beta.
  \end{align}
  Then, $\limsup_{t\to\infty} |V^{-1} x(t)| \preceq \lim_{k\to\infty}
  T_0^k(\beta)$, with $\beta$ as above for $\xi = \max_{-\bar\tau \le
    t \le 0} |V^{-1} x(t)|$.
\end{cor}

The ultimate bounds provided by Corollary~\ref{cor:gub} are
semi-global because every initial condition has an associated ultimate
bound but different initial conditions may produce different ultimate
bounds. 

\subsection{Special case: Affine perturbation bounds}
\label{sec:lin-pert-bnd}

In this subsection, we analyze a specific form of the
bounding function $\delta$ for which global ultimate bounds can be
obtained under a simple sufficient condition. We require the following
preliminary lemma.
\begin{lem}
  \label{lem:RLamF}
  Let $\Lambda \in \R^{n\times n}$ be Metzler, let $\bar
  F\in\R_{+0}^{n\times n}$, and consider 
  \begin{equation}
    \label{eq:4a}
    R \dfn -\Lambda^{-1} \bar F.
  \end{equation}
  Then,
  \begin{enumerate}[a)]
  \item If $\srad(R) < 1$ and $\Lambda$ is Hurwitz, then $\Lambda+\bar F$
    is Hurwitz.\label{item:4}
  \item If $\Lambda+\bar F$ is Hurwitz, then $\Lambda$ is Hurwitz and
    $\srad(R) < 1$.\label{item:5}
  \end{enumerate}
\end{lem}

The main result for the case of affine perturbation bounds is the
following.
\begin{theorem}
  \label{thm:UBlinbnd}
  Consider a switching system (\ref{eq:system}) with perturbation
  bound of the form (\ref{eq:boundw})--(\ref{eq:barx}), where the
  bounding functions $\delta_i$ are CNI. Let $V\in\C^{n\times n}$ be
  invertible, define $\Lambda_i$ and $\Lambda$ as in (\ref{eq:88}),
  and suppose that $\Lambda$ is Hurwitz. Consider $\psi : \R^n_{+0}
  \to \R^n_{+0}$ as defined in (\ref{eq:40}) and suppose that there
  exists
  \begin{equation}
    \label{eq:15a}
    \tilde\delta(x) \dfn \bar F x + \bar w,
  \end{equation}
  for some $\bar F \in \R^{n\times n}_{+0}$ and $\bar w \in
  \R^n_{+0}$, satisfying $\tilde\delta(x) \cge \psi(x)$ for all
  $x\in\R^n_{+0}$, and such that $\srad(R) < 1$ with $R$ as in
  (\ref{eq:4a}).
  Define
  \begin{equation}
    \label{eq:5a}
    \tilde b \dfn (\id - R)^{-1} (-\Lambda^{-1}) \bar w.
  \end{equation}
  Then,
  \begin{enumerate}[(a)]
  \item Invariance. If $|V^{-1}x(t)| \preceq \tilde b$ for $-\bar\tau \le t
    \le 0$, then $|V^{-1}x(t)| \preceq \tilde b$ for all $t\ge
    -\bar\tau$.\label{item:1a}
  \item Global ultimate bounds. $\limsup_{t\to\infty} |V^{-1}x(t)|
    \preceq \tilde b$.\label{item:2a}
  \item Tighter global ultimate bounds. Suppose that there exists a
    continuous and CNI $\delta : \R^n_{+0} \to \R^n_{+0}$
    satisfying
    \begin{equation}
      \label{eq:52}
      \psi(x) \preceq \delta(x) \preceq \tilde\delta(x),\quad
      \text{for all }x\in \R^n_{+0}.
    \end{equation}
    Define $T_0 :\R^n_{+0} \to \R^n_{+0}$ as $T_0(x)= -\Lambda^{-1}
    \delta(x)$. Then, $\limsup_{t\to\infty} |V^{-1}x(t)| \preceq
    \lim_{k\to\infty} T_0^k(\tilde b) \preceq \tilde b$.\label{item:8}
  \item There exists $D$ diagonal and positive definite such that
    \begin{equation}
      \label{eq:14}
      (\Lambda+\bar F)' D + D(\Lambda + \bar F) < 0
    \end{equation}\label{item:10}
  \item Ultimate bounds via standard Lyapunov
    techniques.\label{item:11} If, in addition, $\bar\tau=0$ (no
    delay), then for each $D$ as in (\ref{item:10}) above, the
    derivative\footnote{Strictly mathematically speaking,
        this derivative may not exist at switching instants. This
        problem can be avoided by requiring the switching function to
        be right-continuous and to have a finite number of
        discontinuities in every bounded interval, and by defining
        $\dot L(t,x)$ as an upper Dini derivative. We do not delve
        into these technicalities here.} of the function $L(x) \dfn
    x'P x$ with $P=\real\{(V^{-1})^* D V^{-1}\}$ along any trajectory
    of (\ref{eq:system}) satisfies $\dot L(t,x) < 0$ for all $t$ and
    all $x$ such that $\|x\|$ is big enough.
  \end{enumerate}
\end{theorem}
Theorem~\ref{thm:UBlinbnd} gives an invariant region and global
ultimate bounds for the case when the perturbation bound
$\tilde\delta$ has affine form [see (\ref{eq:15a})]. The main
additional assumption required by this theorem is that the matrix $R$
constructed from the system matrix $\Lambda$ and the perturbation
bound matrix $\bar F$ [see (\ref{eq:4a})] has spectral radius less
than $1$. According to Lemma~\ref{lem:RLamF}\ref{item:5}), we may seek
$V$ causing both $\Lambda$ to be Hurwitz and $\srad(R)<1$ by means of
the following optimization problem, similar to that in
Remark~\ref{rem:opt}:
\begin{quote}
  minimize $\ab(\Lambda+\bar F)$ subject to $V\in\C^{n\times n}$ invertible,
\end{quote}
where it is sufficient to find $V$ so that $\ab(\Lambda+\bar F)<0$.
Note also that, according to the hypotheses of
Theorem~\ref{thm:UBlinbnd} and Lemma~\ref{lem:RLamF}\ref{item:4}), and
since the matrix $\Lambda$ from (\ref{eq:88}) is Metzler for every
$V\in\C^{n\times n}$ invertible, seeking $V$ in the proposed manner
does not incur any loss of generality. We will illustrate this
procedure in Section~\ref{sec:ct-example}.

The main advantage of the affine form of the perturbation bound is
that an invariant region [Theorem~\ref{thm:UBlinbnd}(\ref{item:1a})]
and global ultimate bound [Theorem~\ref{thm:UBlinbnd}(\ref{item:2a})]
can be straightforwardly computed, without having to iterate a map or
to search for a vector $\beta$ such that $T_0(\beta) \prec \beta$ as
was required in Theorem~\ref{thm:nlpertct}: the quantity $\tilde b$ is
guaranteed to exist [under the assumption that $\srad(R)<1$], and can
be computed directly from the expression
(\ref{eq:5a}). 

Theorem~\ref{thm:UBlinbnd}(\ref{item:8}) deals with the
case when the perturbation can be overbounded with affine
$\tilde\delta$ but a tighter CNI perturbation bound $\delta$ exists
which is not of affine form. In this case,
Theorem~\ref{thm:UBlinbnd}(\ref{item:8}) avoids the need to search for
a vector $\beta$ such that $T_0(\beta) \prec \beta$ as in
Theorem~\ref{thm:nlpertct} and shows that a global ultimate bound
possibly tighter than that provided by the quantity $\tilde b$ in
Theorem~\ref{thm:UBlinbnd}(\ref{item:2a}) can be obtained by iterating
the map $T_0$ on $\tilde b$.

Theorem~\ref{thm:UBlinbnd}(\ref{item:10})--(\ref{item:11}) provide a
way of computing a quadratic function so that ultimate bounds can be
obtained via standard Lyapunov techniques, in the case when no delay
is present. Note that how to compute such a suitable quadratic
function is not evident due to the componentwise absolute value in the
form of the perturbation bound (\ref{eq:boundw})--(\ref{eq:barx}).

Results similar to those of Theorem~\ref{thm:UBlinbnd}(\ref{item:2a})
were given in Theorem 3.1 of \cite{kofserhai_auto08} (for
non-switching systems). However, the bounds in the latter reference
require the matrix $V$ to yield the similarity transformation that
takes the system $A$ matrix into Jordan canonical form. Note that
requesting such a condition for $V$ in the current switching case is
usually impossible since not all the different $A_i$ will be taken to
their Jordan canonical form by the same transformation. In addition,
the bounds in Theorem 3.1 of \cite{kofserhai_auto08} are derived
directly on the components of $|x(t)|$ whereas those in
Theorem~\ref{thm:UBlinbnd}(\ref{item:2a}) above correspond to
$|V^{-1}x(t)|$. This difference makes possible the extension of the
ultimate bound results in order to obtain tighter bounds in
Theorem~\ref{thm:UBlinbnd}(\ref{item:8}) and to derive the
relationship with CQLF in
Theorem~\ref{thm:UBlinbnd}(\ref{item:10})--(\ref{item:11}).

\begin{rem}
  \label{rem:suffcond}
  If the constant part $\bar w$ of the affine bound (\ref{eq:15a}) is
  zero, then $\tilde b=0$ in (\ref{eq:5a}) and
  Theorem~\ref{thm:UBlinbnd}(\ref{item:2a}) implies that
  $\lim_{t\to\infty} x(t) = 0$. Consequently, the condition
  $\srad(R)<1$ in Theorem~\ref{thm:UBlinbnd} or, equivalently
  according to Lemma~\ref{lem:RLamF}, the condition $\Lambda+\bar F$
  Hurwitz, is a sufficient condition for the uniform stability of a
  switching system with a perturbation bound depending linearly on the
  componentwise absolute value of a delayed state.
\end{rem}

In the following section we illustrate all the above results by means
of a numerical example.

\section{Example}
\label{sec:ct-example}

Consider a switching system of the form (\ref{eq:system}), with $N=2$,
$n=3$, $k_1=1$, $k_2=2$, and
\begin{align}
  \label{eq:17a}
  A_1 &=
  \begin{bmatrix}
    -6.91 & 1.92 & 4.4\\
    1.32  & -1.54 & -1.41\\
    4.47  & -3.02 & -5.43
  \end{bmatrix},
  & H_1 &=
  \begin{bmatrix}
    0\\ 0.02\\ 0
  \end{bmatrix},\displaybreak[0]\\
  \label{eq:24a}
  A_2 &=
  \begin{bmatrix}
    -9.27 & -0.19 & 7.15\\
     2.02 & -1.38 & -1.94\\
     6.84 & -4.28 & -6.64
  \end{bmatrix},
  & H_2 &=
  \begin{bmatrix}
    .01 & -.05\\
    .01 & 0\\
    .02 & .03
  \end{bmatrix}.
\end{align}
The perturbation vectors $w_1(t)\in\R$ and $w_2(t)\in\R^2$ are
componentwise bounded by $|w_i(t)| \cle \hat\delta_i(\theta(t))$ with
$\theta(t)$ as defined in (\ref{eq:barx}), $\bar\tau=0.1$,
$\hat\delta_1 : \R^3_{+0} \to \R^1_{+0}$ and $\hat\delta_2 : \R^3_{+0}
\to \R^2_{+0}$ given by
\begin{align}
  \label{eq:18a}
  \hat\delta_1(\theta) &= |\sin \theta_3|,\\
  \label{eq:45}
  \hat\delta_2(\theta) &=
  \begin{bmatrix}
    \theta_1 e^{-2\theta_1} + |\cos \theta_2|\\
    5\theta_3 + 1
  \end{bmatrix},
\end{align}
and $\theta = [\theta_1,\theta_2,\theta_3]'$. Note that $\hat\delta_1$
and $\hat\delta_2$ as in (\ref{eq:18a})--(\ref{eq:45}) are
continuous but not CNI. Following Remark~\ref{rem:cni},
we compute the tightest continuous CNI overbounds
$\delta_1$ and $\delta_2$:
\begin{align}
  \label{eq:59}
  \hat\delta_1(\theta) &\preceq \delta_1(\theta) \dfn
  \begin{cases}
    \sin\theta_3 &\text{if }0 \le \theta_3 \le \pi/2,\\
    1 &\text{if }\theta_3 \ge \pi/2.
  \end{cases}\\
  \label{eq:61}
  \hat\delta_2(\theta) &\preceq \delta_2(\theta) \dfn
  \begin{bmatrix}
    \begin{cases}
      \theta_1 e^{-2\theta_1}+1 &\text{if }\theta_1 \le 1/2,\\
      e^{-1}/2+1 &\text{if }\theta_1 > 1/2,
    \end{cases}\\
    5\theta_3 + 1
  \end{bmatrix}.
\end{align}
In turn, $\delta_1$ and $\delta_2$ have affine bounds, as we next
show. From (\ref{eq:59})--(\ref{eq:61}), we have
\begin{align}
  \label{eq:48}
  \delta_1(\theta) &\le \theta_3 &= \bar F_1 \theta + \bar w_1,\\
  \label{eq:49}
  \delta_2(\theta) &\preceq
  \begin{bmatrix}
    \theta_1 + 1\\
    5\theta_3 + 1
  \end{bmatrix} &= \bar F_2 \theta + \bar w_2,
\end{align}
where we have defined
\begin{align}
  \label{eq:55}
  \bar F_1 &\dfn
  \begin{bmatrix}
    0 & 0 & 1
  \end{bmatrix},
  &\bar w_1 &\dfn 0,\\ 
  \label{eq:23}
  \bar F_2 &\dfn
  \begin{bmatrix}
    1 & 0 & 0\\
    0 & 0 & 5
  \end{bmatrix},
  &\bar w_2 &\dfn
  \begin{bmatrix}
    1\\ 1
  \end{bmatrix}.
\end{align}

\subsection{Nonlinear perturbation bound}
\label{sec:nonl-pert-bound}

\subsubsection{Transient and ultimate bounds via componentwise method}

In order to apply Theorem~\ref{thm:nlpertct}, we need to find a
suitable invertible matrix $V$ and a positive vector $\beta$ so that
$T_0(\beta) \prec \beta$. To find an invertible $V\in\C^{n\times n}$
such that $\Lambda$ in (\ref{eq:88}) is Hurwitz, we follow the
strategy outlined in Remark~\ref{rem:opt}. We thus minimize
$\ab(\Lambda)$ searching over $V$. This optimization was implemented
in Matlab$^{\textregistered}$, yielding
\begin{equation}
  \label{eq:4}
  V = \left[
    \begin{smallmatrix}
      2.408 & 1.745 & 0.162\\
      -0.634 & -1.363 & 0.0351\\
      -2.144 & 2.217 & 0.118
    \end{smallmatrix}\right] +
    \left[
      \begin{smallmatrix}
        .443 & 2.059 & 1.558\\
        -.117 & -1.815 & .494\\
        -.399 & 3.247 & 1.652
      \end{smallmatrix}\right]\mathbf{i}
\end{equation}
for which, from (\ref{eq:88}),
\begin{equation}
  \label{eq:31}
  \Lambda =
  \begin{bmatrix}
    -11.34 & 1.145 & .191\\
    .0067 & -.0979 & .0038\\
    .0130 & 1.912 & -1.605
  \end{bmatrix}
\end{equation}
and $\ab(\Lambda) = -.0923 < 0$. Next, we require a continuous and CNI
function $\delta$ satisfying (\ref{eq:6}). Since both $\delta_1$ and
$\delta_2$ in (\ref{eq:59})--(\ref{eq:61}) are continuous and CNI,
then $\psi$ as defined in (\ref{eq:40}) is continuous and CNI, and
hence we may take $\delta\equiv\psi$. We next follow the procedure
given in Algorithm~1 and Theorem~3 of \cite{kofhai_ijc07} in order to
find the positive vector $\beta$. Using this procedure we select a
positive vector $\alpha=[1,1,1]'$ and iterate $T_\alpha(x) =
-\Lambda^{-1}\psi(x) + \alpha$ from $0$, numerically computing $\beta
= \lim_{k\to\infty} T_\alpha^k(0)$, for which
\begin{equation}
  \label{eq:63}
  T_0(\beta) =
  \begin{bmatrix}
    3.235\\
    18.23\\
    25.82
  \end{bmatrix} \prec \beta =
  \begin{bmatrix}
    4.235\\
    19.23\\
    26.82
  \end{bmatrix}.
\end{equation}
By Theorem~\ref{thm:nlpertct}(\ref{item:5a}), we can (numerically)
compute
\begin{equation}
  \label{eq:64}
  b = \lim_{k\to\infty} T_0^k(\beta) = 
  \begin{bmatrix}
    0.127\\
    0.715\\
    1.017
  \end{bmatrix},\quad |V|b =
  \begin{bmatrix}
    3.84\\
    2.21\\
    4.78
  \end{bmatrix}.
\end{equation}
Then, application of Theorem~\ref{thm:nlpertct}(\ref{item:6}) with
$\gamma=0$ shows that if $|V^{-1}x(t)| \preceq T_0(\beta)$ for all
$-\bar\tau \le t \le 0$, then $|V^{-1}x(t)| \preceq \beta$ for all
$t\ge -\bar\tau$, and the combination of the latter result with
Theorem~\ref{thm:nlpertct}(\ref{item:7}) shows that if $|V^{-1} x(t)|
\preceq T_0(\beta)$ for all $-\bar\tau \le t \le 0$, then
$\limsup_{t\to\infty} |V^{-1}x(t)| \preceq b$, with $b$ given by
(\ref{eq:64}) and, according to Remark~\ref{rem:shape}, also
$\limsup_{t\to\infty} |x(t)| \cle |V|b$.

In addition, if we require a larger set of initial conditions for
which the bounds should be valid, we may follow
Theorem~\ref{thm:nlpertct}(\ref{item:9}). We thus select $c=[1,1,1]'$
and compute $\bar\epsilon=0.8384$, according to (\ref{eq:67}). Consequently,
the transient bounds $|V^{-1}x(t)| \preceq \beta$ for all $t\ge
-\bar\tau$ will be valid not only if $|V^{-1}x(t)| \preceq T_0(\beta)$
for all $-\bar\tau \le t \le 0$ but also whenever $|V^{-1}x(t)|
\preceq T_\gamma(\beta) = T_0(\beta) + \gamma$ for all $-\bar\tau \le t \le 0$, with
$\gamma=c\epsilon$ and any positive $\epsilon<\bar\epsilon$. For
example, for $\epsilon=0.838<\bar\epsilon$, then 
$T_\gamma(\beta) = [
  4.073, 19.068, 26.658]'.$

\subsubsection{Ultimate bound via quadratic Lyapunov function}

We next intend to compute ultimate bounds by means of a quadratic
Lyapunov function. Note that the matrix $V$ in (\ref{eq:4}) was
obtained using information on only the switching linear part of the
system, without information on the perturbation bound. Also, note that
the bounds computed above by means of Theorem~\ref{thm:nlpertct} are
the same for every value of the maximum delay, $\bar\tau$, provided
that the bound on the initial condition is satisfied for all
$-\bar\tau \le t \le 0$. In order to derive ultimate bounds by means
of a Lyapunov function, we next assume that $\bar\tau = 0$.

The computation of a quadratic Lyapunov function $L(x) = x'Px$ for the
switching linear part of the system (disregarding the perturbation)
can be performed via solving the LMIs
\begin{equation}
  \label{eq:25}
  A_i' P + P A_i < 0, \quad\text{for }i\in\Ind,
\end{equation}
with $P=P' > 0$. Solving these LMIs in Matlab$^{\textregistered}$ yields
\begin{equation}
  \label{eq:26}
  P =
  \begin{bmatrix}
    .1638 & .1634 & .012\\
    .1634 & 1.9577 & -.3602\\
    .012  & -.3602 & .2285
  \end{bmatrix}
\end{equation}
The derivative of $L(x)$ along the trajectories of the system satisfies
\begin{align}
  \dot L(t,x) &= x' (A_{\sw(t)}' P + P A_{\sw(t)}) x + 2x' P H_{\sw(t)} w_{\sw(t)}(t)\notag\\
  &\le \max_{i\in\Ind} \left[x' (A_i' P + P A_i) x + 2 \max_{|w| \cle
      \delta_i(|x|)} |x'PH_i w|\right]\notag\\
  \label{eq:29}
  &= \max_{i\in\Ind} \left[x' (A_i' P + P A_i) x + 2  |x'PH_i| \delta_i(|x|)\right]
\end{align}
Note that the bound on $\dot L(t,x)$ given by (\ref{eq:29}) is tight,
i.e., for every $x\in\R^n$, there exists a switching state $\sw(t)$
and a possible value of $w_{\sw(t)}(t)$ so that $\dot L(t,x)$ equals
the right-hand side of (\ref{eq:29}).  A necessary condition to be
able to compute an ultimate bound by means of $L(x)$ is that
$\max_{x'Px=k} \dot L(t,x) < 0$ for some $k>0$. Numerical search for
such a $k>0$ yields no solution.

An alternative way of computing a quadratic Lyapunov function without
employing information on the perturbation bound is given by
Theorem~\ref{thm:cdlf}\ref{item:3}) using $\bar\Lambda = \Lambda$. We
thus solve the LMIs (\ref{eq:2}) for $D>0$ diagonal. This yields
$D=\diag(.0411,.5584,.0800)$ for which
\begin{equation*}
  P = \real\{(V^{-1})^* D V^{-1}\} =
  \begin{bmatrix}
    .0088 & .0134 &.0017\\
    .0134 & .0763 &-.0063\\
    .0017 & -.0063 & .0074
  \end{bmatrix}
\end{equation*}
As with the previous $P$ above, numerical search for $k>0$ such that
$\max_{x'Px=k} \dot L(t,x) < 0$ yields no solution.

\subsection{Affine perturbation bound}
\label{sec:line-pert-bound}

\subsubsection{Ultimate bound via componentwise method}

We next will take the affine perturbation bound into account for the
computation of the matrix $V$. Since the perturbation bounds
$\delta_1$ and $\delta_2$ admit affine bounds, as shown by
(\ref{eq:48})--(\ref{eq:23}), then the function $\psi$ in
(\ref{eq:40}) corresponding to $\delta_1$ and $\delta_2$ as in
(\ref{eq:59})--(\ref{eq:61}) can actually be bounded by an affine CNI
function $\tilde\delta$ for every $V\in\C^{n\times n}$
invertible. 
To see this, note that
\begin{equation}
  \label{eq:22a}
  \max_{|w_i| \preceq \delta_i(|V|x)} |V^{-1} H_i w_i| \preceq |V^{-1} H_i| \delta_i(|V|x),
\end{equation}
for $i=1,2$ [note that the right-hand side of (\ref{eq:22a}) may not
be a tight bound on its left-hand side only when $V$ has complex
components]. Combining (\ref{eq:48})--(\ref{eq:49}) and
(\ref{eq:22a}), and recalling (\ref{eq:40}), we have
\begin{align}
  \label{eq:56}
  \psi(x) &\preceq \max_{i\in\{1,2\}} \left[ |V^{-1} H_i| (\bar F_i |V| x
    + \bar w_i) \right]\\
  \label{eq:57}
  &\preceq \tilde\delta(x) \dfn \bar F x + \bar w,\quad\text{with}\\
  \label{eq:Fbar}
  \bar F &\dfn \max_{i\in\{1,2\}} |V^{-1}H_i| \bar F_i |V|,\\
  \label{eq:3a}
  \bar w &\dfn \max_{i\in\{1,2\}} |V^{-1}H_i| \bar w_i.
\end{align}
We have thus shown that for each $V\in\C^{n\times n}$ invertible, the
nonnegative function $\psi$ in (\ref{eq:40}) admits a bound
$\tilde\delta$ of the affine form (\ref{eq:15a}). In order to apply
Theorem~\ref{thm:UBlinbnd}, we require an invertible matrix $V$ so
that $\Lambda$ is Hurwitz and $\srad(R)<1$, with $R$ as in
(\ref{eq:4a}). The previously used matrix $V$ given in~\eqref{eq:4}
does not satisfy $\srad(R)<1$, hence a new $V$ is required. According
to Lemma~\ref{lem:RLamF}\ref{item:5}), it suffices to find $V$ such
that $\Lambda+\bar F$ is Hurwitz. Similarly to Remark~\ref{rem:opt},
we seek $V$ by means of the following optimization problem:
\begin{quote}
  minimize $\ab(\Lambda+\bar F)$ subject to $V\in\C^{n\times n}$ invertible,
\end{quote}
where it is sufficient to find $V$ such that $\ab(\Lambda+\bar F) < 0$.
Performing this optimization in Matlab$^{\textregistered}$ yields
\begin{equation}
  \label{eq:20a}
  V =
  \left[\begin{smallmatrix}
    2.244 & -2.715 & 0\\
    0.706 & 0.715 & -4.302\\
    2.359 & 2.418 & 1.674
  \end{smallmatrix}\right] +
  \left[\begin{smallmatrix}
    -4.401 & 2.891 & 0\\
    -1.385 & -0.761 & -3.789\\
    -4.625 & -2.575 &1.470
  \end{smallmatrix}\right]\mathbf{i}.
\end{equation}
Operating as in (\ref{eq:88}) yields
\begin{equation}
  \label{eq:21a}
  \Lambda =
  \begin{bmatrix}
    -1.599 & 0.001 & 2.620\\
    0.268  & -11.34 & 1.028\\
    0.006  & 0.004  & -0.103
  \end{bmatrix},
\end{equation}
and, from (\ref{eq:Fbar})--(\ref{eq:3a}),
\begin{equation*}
  \bar F =
  \begin{bmatrix}
    0.633 & 0.450 & 0.205\\
    2.749 & 1.879 & 1.150\\
    0.390 & 0.269 & 0.154
  \end{bmatrix}\cdot 10^{-1},\quad
  \bar w =
  \begin{bmatrix}
    0.5\\ 1.17\\ 0.2
  \end{bmatrix}\cdot 10^{-2}.
\end{equation*}
Computing the matrix
$R$ in (\ref{eq:4a}) yields $\srad(R)<1$ and we may obtain $\tilde b$
as in (\ref{eq:5a}):
\begin{equation}
  \label{eq:58}
  \tilde b =
  \begin{bmatrix}
    0.903\\
    0.098\\
    0.521
  \end{bmatrix},\qquad
  |V|\tilde b =
  \begin{bmatrix}
    4.85\\
    4.49\\
    6.20
  \end{bmatrix}.
\end{equation}
By Theorem~\ref{thm:UBlinbnd}(\ref{item:2a}) we have
$\limsup_{t\to\infty} |V^{-1}x(t)| \preceq \tilde b$, and according to
Remark~\ref{rem:shape}, then $\limsup_{t\to\infty} |x(t)| \cle
|V|\tilde b$, for every initial condition. It may be surprising that
the componentwise ultimate bound $|V|\tilde b$ in (\ref{eq:58}) is
more conservative than the corresponding one in
(\ref{eq:64}). However, the current bounds are valid from every
initial condition as opposed to the ones in
Section~\ref{sec:nonl-pert-bound}.
%
%
%
%
In addition, we may seek a global ultimate bound tighter than the one
corresponding to $\tilde b$ in (\ref{eq:58}) by applying
Theorem~\ref{thm:UBlinbnd}(\ref{item:8}). Note that if we take $\delta
= \psi$, with $\psi$ as in (\ref{eq:40}) for $\delta_1$ and $\delta_2$
as in (\ref{eq:59})--(\ref{eq:61}), then the functions $\delta=\psi$
and $\tilde\delta$ as in (\ref{eq:57})--(\ref{eq:3a}) satisfy
(\ref{eq:52}). Thus, we iterate the map $T_0(x) = -\Lambda^{-1}
\delta(x)$ on the vector $\tilde b$ computed in (\ref{eq:58}). This
yields
\begin{equation}
  \label{eq:62}
  b = \lim_{k\to\infty} T_0^k(\tilde b) =
  \begin{bmatrix}
    0.365\\
    0.0403\\
    0.212
  \end{bmatrix},\quad
  |V| b =
  \begin{bmatrix}
    1.96\\
    1.82\\
    2.51
  \end{bmatrix},
\end{equation}
which are clearly tighter than those in (\ref{eq:58}). Moreover, the
componentwise ultimate bound $|V|b$ in (\ref{eq:62}) is also tighter
than the corresponding one in (\ref{eq:64}).

\subsubsection{Ultimate bound via quadratic Lyapunov function}

According to Theorem~\ref{thm:UBlinbnd}(\ref{item:11}), we may compute a quadratic
Lyapunov function $L(x)=x'Px$ suitable for the obtention of ultimate
bounds by means of the matrix $V$ computed in (\ref{eq:20a}). We thus
solve the LMIs (\ref{eq:14}) for $D>0$ diagonal, yielding
$D=\diag(.1812,.5127,9.962)$ and
\begin{equation*}
  P = \real\{(V^{-1})^* D V^{-1}\} =
  \begin{bmatrix}
    .0111 & -.003 & -.0064\\
    -.003 & .245 & -.0692\\
    -.0064 & -.0692 & .0301
  \end{bmatrix}
\end{equation*}
Numerical computation of the smallest $k>0$ for which $\max_{x'Px=k}
\dot L(t,x) < 0$ yields $k=.0989$, from which it can be verified that
$\dot L(t,x) < 0$ for all $x$ satisfying $x'Px > .0989$. Therefore,
$\limsup_{t\to\infty} x(t)'Px(t) \le .0989$ and we may compute the
componentwise bounds $\bar x_i = \max_{x'Px=.0989} x_i$ for
$i=1,2,3$:
\begin{equation}
  \label{eq:50}
  \begin{bmatrix}
    \bar x_1\\ \bar x_2\\ \bar x_3
  \end{bmatrix} =
  \begin{bmatrix}
    4.0448\\
    0.7926\\
    4.1443
  \end{bmatrix}
\end{equation}
The bounds for $x_1$ and $x_3$ are more conservative than the
corresponding ones given by (\ref{eq:62}) but the bound for $x_2$ is
tighter. Note that this tighter bound on the second component of the
state vector would be completely lost if bounds on the 1, 2 or
$\infty$ norms were obtained based on the fact that
$\limsup_{t\to\infty} x(t)'Px(t) \le .0989$. The bounds (\ref{eq:62})
and (\ref{eq:50}) may be combined, yielding a global ultimate bound
better than either one:
\begin{equation}
  \label{eq:65}
  \limsup_{t\to\infty}|x(t)| \cle
  \begin{bmatrix}
    1.96\\
    .7926\\
    2.51
  \end{bmatrix}
\end{equation}

\section{Conclusions}
\label{sec:conc}

We have proposed a method to compute componentwise transient bounds,
componentwise ultimate bounds, and invariant regions for a class of
switching continuous-time linear systems with perturbation bounds that
may depend nonlinearly on a delayed state. We have provided conditions
for the bounds to be of local or semi-global nature. We have also
addressed the particular case of perturbation bounds that have affine
dependence on a delayed state, for which the bounds derived are shown
to be of global nature and a novel sufficient condition for practical
stability was provided.
Another contribution of the paper was to establish that the class of
switching linear systems to which our componentwise bound and
invariant set method can be applied is strictly contained in the class
of switching linear systems that admit a CQLF, although the switching
linear system need not be close to simultaneously
triangularizable. This closes a problem left open in our previous
paper \cite{haiser_auto10}. A third contribution was to provide a
technique to compute a CQLF for switching linear systems with
perturbations bounded componentwise by affine functions of the
absolute value of the state vector components (when no delay is
present).
Future work may focus on switched systems where either the switching
signal or a continuous control input can be designed in order to
ensure a given ultimate bound (cf. \cite{kofserhai_auto08}) and on the
extension and application of the current results to networked control
systems (cf. \cite{haikof_auto07}) and to switching systems with mixed
continuous- and discrete-time dynamics.

\appendix

\section{Proofs}
\label{sec:proofs}

\subsection{Preliminary Lemmas}
\label{sec:preliminary-lemmas}

The following two lemmas derive properties of CNI functions that are
required in the proof of our main results.

\begin{lem}\label{lem:CNI}
  Let $f : \R^n_{+0} \to \R^n_{+0}$ be a continuous CNI function and
  suppose that there exists $\beta \in \R^n_{+0}$ satisfying $f(\beta)
  \preceq \beta$. Then:
  \begin{enumerate}[(i)]
  \item For every $k\in\integs_{+}$, $f^{k+1}(\beta) \preceq
    f^k(\beta)$ and 
    \begin{equation}
      \label{eq:b}
      \lim_{k\to\infty} f^k(\beta) = b \succeq0.
    \end{equation}\label{item:CNIa}
  \item For every $\epsilon\in\R^n_+$ there exist $k=k(\epsilon) \in
    \integs_+$ and $\gamma=\gamma(\epsilon) \in \R^n_+$ such that
    $f_{\gamma}^{k}(\beta) \prec b + \epsilon$, where $b$ is as
    in~\eqref{eq:b} and $f_\gamma(x) \triangleq f(x)+\gamma$, $\forall
    x\in \R^n_{+0}$.\label{item:CNIb}
 \end{enumerate}
\end{lem}
\begin{proof}
  (i) Applying the CNI property to the inequality $f(\beta) \preceq
  \beta$ and iterating the process, it follows that $f^{k+1}(\beta)
  \preceq f^k(\beta)$ for all $k\in\integs_{+}$. Also, since $f$ maps
  nonnegative vectors to nonnegative vectors, then $f^k(\beta) \succeq
  0$ for all $k\in\integs_{+}$. It follows that the vectors
  $f^k(\beta)$ form a componentwise nonincreasing sequence which is
  lower bounded by 0. Hence, each component must converge to some
  nonnegative real number and thus~\eqref{eq:b} holds.

  (ii) Note that $|f_{\gamma}^{k}(\beta) - b| \preceq
  |f_{\gamma}^{k}(\beta) - f^{k}(\beta)| + |f^{k}(\beta) -
  b|$. From~\eqref{eq:b}, given $\epsilon \in \R^n_{+}$, we can select
  $k=k(\epsilon)$ such that $|f^{k}(\beta) - b| \prec
  \epsilon/2$. From the definition of $f_\gamma$ and the continuity of
  $f$, it follows that, for the selected value of $k$, we may select
  $\gamma = \gamma(\epsilon) \in \R^n_{+}$ small enough so that
  $|f_{\gamma}^{k}(\beta) - f^{k}(\beta)| \prec \epsilon/2$. Then,
  $|f_{\gamma}^{k}(\beta) - b| \prec \epsilon$, whence
  $f_{\gamma}^{k}(\beta) \prec b + \epsilon$.
\end{proof}

\begin{lem}
  \label{lem:affine}
  Consider the affine function $\ell(x) \triangleq R x +r$ where $r
  \in \R^n_{+0}$ and $R \in \R^{n\times n}_{+0}$ is such that
  $\srad(R) < 1$. Then:
  \begin{enumerate}[(i)]
  \item The function $\ell : \R^n_{+0} \to \R^n_{+0}$ is CNI.\label{item:ellCNI}
  \item For all $\beta \in \R^n_{+0}$, $\lim_{k\to\infty}
    \ell^k(\beta) = \tilde{b} =\ell(\tilde b)$, where
    \begin{equation}
      \label{eq:btilde}
      \tilde{b}= (I-R)^{-1} r.
    \end{equation}\label{item:lim}
  \item For every $v \in \R^n_{+0}$ there exists $\beta \in \R^n_{+}$
    satisfying
    \begin{equation}
      \label{eq:v}
      \ell(\beta) + v \prec \beta.
    \end{equation}\label{item:LINb}
  \item For every $\epsilon \in \R^n_+$, there exists $\beta_\epsilon
    \in \R^n_{+0}$ satisfying $\tilde b \preceq \beta_\epsilon \preceq
    \tilde b+\epsilon$ and $ \ell(\beta_\epsilon) \prec
    \beta_\epsilon$, with $\tilde b$ as in
    (\ref{eq:btilde}).\label{item:LINc}
  \item Let $f : \R^n_{+0} \to \R^n_{+0}$ be a continuous CNI function
    satisfying $f(x) \preceq \ell(x)$ for all $x\in\R^n_{+0}$.  Let
    $\beta \in \R^n_{+}$ be such that \eqref{eq:v} holds for some
    $v\in\R^n_{+0}$, and let $\tilde{b}$ be as in~\eqref{eq:btilde}.
    Then~\eqref{eq:b} holds and, in addition,
    \begin{equation}
      \label{eq:66}
      b= \lim_{k\to\infty} f^k(\tilde b) \preceq \tilde{b}.
    \end{equation}\label{item:LINd}
  \end{enumerate}
\end{lem}
\begin{proof}
  By assumption we have $R\succeq 0$ and $\srad(R)<1$. Let
  $R_\epsilon$ be a slight perturbation of $R$ so that $R_\epsilon
  \succ R$ and $\srad(R_\epsilon)<1$. Then, $R_\epsilon \succ 0$ and
  by the Perron-Frobenius Theorem (see, e.g. Theorem~8.2.2 of
  \cite{horjoh_book85}) then $\srad(R_\epsilon) > 0$ and there exists
  $x\succ 0$ such that $R_\epsilon x = \srad(R_\epsilon) x$. It
  follows that
  \begin{equation}
    \label{eq:26a}
    Rx \prec R_\epsilon x = \srad(R_\epsilon) x \prec x.
  \end{equation}
  
  (\ref{item:ellCNI}) Immediate from the fact that $R \succeq 0$.

  (\ref{item:lim}) Immediate from the assumption $\srad(R)<1$.
  
  (\ref{item:LINb}) From \eqref{eq:26a}, $y \dfn (\id - R)x \succ 0$. Define $z \dfn
    r+v$ and let $y_i$ and $z_i$ denote the $i$-th components of $y$
    and $z$, respectively. Select $\alpha > 0$ so that
    \begin{equation}
      \label{eq:27a}
      \alpha > \max_{i=1,\ldots,n} \left\{ \frac{z_i}{y_i} \right\}
    \end{equation}
    and define $\beta\dfn \alpha x$. Note that $\beta \succ 0$. Then,
    $\alpha y =(\id-R)\beta \succ r+v$. Operating on the latter
    inequality yields $R \beta+r+v =\ell(\beta) +v\prec
    \beta$, and the result follows.
  
    (\ref{item:LINc}) Since $x \succ 0$, for every $\alpha > 0$ we have $\tilde
    b+\alpha x \succ \tilde b$ and
    \begin{equation}
      \label{eq:30a}
      \ell(\tilde b+\alpha x) = R(\tilde b+\alpha x) +r  
      = \tilde b + \alpha Rx,
    \end{equation}
    where we have used \eqref{eq:btilde}. From \eqref{eq:26a} and
    \eqref{eq:30a}, it follows that $\ell(\tilde b+\alpha x)
    \prec \tilde b+\alpha x$ for every $\alpha > 0$. Given $\epsilon
    \in \R^n_+$, select $\alpha_\epsilon$ satisfying
    \begin{equation}
      \label{eq:39}
      0 < \alpha_\epsilon \le \min_{i=1,\ldots,n} 
      \left\{ \frac{\epsilon_i}{x_i} \right\}
    \end{equation}
    and define $\beta_\epsilon = \tilde b + \alpha_\epsilon x$. Then,
    $\tilde b \preceq \beta_\epsilon \preceq \tilde b+\epsilon$ and
    $\ell (\beta_\epsilon) \prec \beta_\epsilon$, establishing
    part~(\ref{item:LINc}).
  
    (\ref{item:LINd}) Note that \eqref{eq:v} with $v\succeq 0$ implies
    $\ell (\beta) \prec \beta$.  We then have $f(\beta) \preceq \ell
    (\beta) \prec \beta$. Also, by
    Lemma~\ref{lem:CNI}(\ref{item:CNIa}), then \eqref{eq:b} holds.
    Since both $f$ and $\ell$ are CNI and $f(x) \preceq \ell(x)$ for
    all $x\in\R^n_{+0}$, then $f^k(\beta) \preceq \ell^k(\beta) \prec
    \beta$ for all $k\in\integs_+$, whence applying limits yields $b
    \preceq \tilde b \prec \beta$. Applying the CNI property of $f$ to
    the latter inequalities, and iterating, yields $b = f^k(b) \preceq
    f^k(\tilde b) \preceq f^k(\beta)$, whence $b \preceq
    \lim_{k\to\infty} f^k(\tilde b) \preceq b$. 
    We have thus established~\eqref{eq:66}.
\end{proof}

\subsection{Proof of Theorem~\ref{thm:cdlf}}
\label{sec:proof-theorem}

  \ref{item:1}) Since $\Lambda$ is Metzler and
  $\bar\Lambda\cge\Lambda$, then $\bar\Lambda$ also is Metzler. Since
  $\bar\Lambda$ is then Metzler and Hurwitz, it admits a diagonal
  Lyapunov function (see, e.g. \cite[Ch.6]{berple_book94}).

  \ref{item:2}) Since $\bar\Lambda$ is Metzler and $D$ is diagonal
  with positive main-diagonal entries, then
  $\bar\Lambda'D+D\bar\Lambda$ is Metzler and symmetric. Combining the
  latter fact with Lemma~\ref{lem:propMetzler}d) and (\ref{eq:2}),
  then
  \begin{equation}
    \label{eq:3}
    x' (\bar\Lambda'D+D\bar\Lambda) x \le |x|'(\bar\Lambda'D+D\bar\Lambda)|x| < 0
  \end{equation}
  for all nonzero $x\in\R^n$. 
  Since $\Lambda \cle \bar\Lambda$, then $\Lambda'D+D\Lambda \cle
  \bar\Lambda'D+D\bar\Lambda$ and hence
  \begin{equation}
    \label{eq:17}
    |x|' (\Lambda'D+D\Lambda) |x| \le |x|'
    (\bar\Lambda'D+D\bar\Lambda) |x|
  \end{equation}
  for all $x\in\R^n$. 
  Combining (\ref{eq:3})--(\ref{eq:17}) and
  Lemma~\ref{lem:propMetzler}d), then 
  \begin{equation}
    \label{eq:5}
    \Lambda' D + D\Lambda < 0.
  \end{equation}
  This establishes that $\Lambda$ is Hurwitz.

  \ref{item:3}) Since $\Lambda$ satisfies (\ref{eq:5}) and by (\ref{eq:88}) $M_i
  \cle \Lambda$ and are Metzler, arguments identical to those in the
  proof of part \ref{item:2}) above show that
  \begin{equation}
    \label{eq:7}
    M_i'D+DM_i < 0,\qquad \text{for all }i\in\Ind.
  \end{equation}
  By (\ref{eq:88}) and since $D$ is diagonal with positive
  main-diagonal entries, then $\Met(\Lambda_i^* D + D\Lambda_i)\cle
  M_i'D+DM_i$. The latter fact implies that
  \begin{equation}
    \label{eq:8}
    |z|' \Met(\Lambda_i^* D + D\Lambda_i) |z| \le |z|'(M_i'D+DM_i)|z|
  \end{equation}
  for all $z\in\C^n$. By Lemma~\ref{lem:propMetzler}e) and
  combining with (\ref{eq:7})--(\ref{eq:8}), it follows that
  \begin{equation}
    \label{eq:21}
    z^* (\Lambda_i^* D + D\Lambda_i) z < 0
  \end{equation}
  for all nonzero $z \in \C^n$. 
  Therefore $\Lambda_i^* D + D \Lambda_i <
  0$ and hence, using (\ref{eq:88}), then $V^* A_i' (V^{-1})^* D + D
  V^{-1} A_i V < 0$. Left-multiplying by $(V^{-1})^*$ and
  right-multiplying by $V^{-1}$ yields $A_i' (V^{-1})^* D V^{-1} +
  (V^{-1})^* D V^{-1} A_i < 0$, whence $A_i' P + P A_i < 0$.

\subsection{Proof of Theorem~\ref{thm:nlpertct}}
\label{sec:proof-theorem-1}

  (\ref{item:5a}) Since $-\Lambda^{-1} \succeq 0$ (see Lemma~\ref{lem:propMetzler})
  and $\delta$ is CNI, then the maps $T_\gamma$ defined
  in~\eqref{eq:19} are CNI for every $\gamma \in
  \R^n_{+0}$. Part~(\ref{item:5a}) then follows by applying
  Lemma~\ref{lem:CNI}(\ref{item:CNIa}) with~$f=T_0$.

  (\ref{item:6}) Since $-\Lambda^{-1} \succeq 0$, then
  \begin{equation}
    \label{eq:8a}
    \gamma = (-\Lambda^{-1})(-\Lambda)\gamma \preceq -\Lambda^{-1} \max\{ -\Lambda\gamma, 0\}.
  \end{equation}
  Adding $-\Lambda^{-1}\delta(\beta)$ to each side of the inequality
  (\ref{eq:8a}), recalling (\ref{eq:19}), and using the assumption,
  yields
  \begin{equation}
    \label{eq:10}
    T_\gamma(\beta) \preceq -\Lambda^{-1} [\delta(\beta) + \max\{ -\Lambda\gamma, 0\}] \prec \beta.
  \end{equation}
  Let $t_c$ be the largest time instant for which
  \begin{equation}
    \label{eq:34}
    |V^{-1} x(t)| \preceq \beta,\quad \text{for all } -\bar\tau
    \le t \le t_c.
  \end{equation}
  Note that $t_c > 0$ necessarily since $|V^{-1} x(t)| \preceq
  T_\gamma(\beta)$ for all $-\bar\tau \le t \le 0$ by assumption, and
  $T_\gamma(\beta) \prec \beta$ by (\ref{eq:10}). It follows from
  (\ref{eq:barx}) that
  \begin{align}
    \label{eq:25a}
    \theta(t) &= \max_{t-\bar\tau \le \tau \le t} |VV^{-1}x(\tau)|\\
    \label{eq:23a}
    &\preceq \max_{t-\bar\tau \le \tau \le t} |V| |V^{-1}x(\tau)| \preceq |V|\beta,
  \end{align}
  for all $0 \le t \le t_c$. From~(\ref{eq:25a})--(\ref{eq:23a}) and
  since $\delta_i$ are CNI, then $\delta_i(\theta(t)) \cle
  \delta_i(|V|\beta)$ for all $0 \le t \le t_c$. Recalling
  (\ref{eq:boundw}), then $|w_i(t)| \preceq \delta_i(|V| \beta)$ for
  all $0 \le t \le t_c$. Define
  \begin{align}
    \label{eq:35}
    \bw_i &\dfn \delta_i(|V| \beta),\\
    \label{eq:37}
    \rb &\dfn \delta(\beta) + \max \big\{ - \Lambda \gamma, 0 \big\},
  \end{align}
  with $\delta$ satisfying (\ref{eq:6}), and note that by
  (\ref{eq:10}), then $T_\gamma(\beta) \preceq -\Lambda^{-1} \rb \prec
  \beta$. Combining the latter inequality with the assumption on the
  initial condition, it follows that $|V^{-1} x(0)| \preceq
  T_\gamma(\beta) \preceq -\Lambda^{-1} \rb$, whence $|V^{-1} x(0)| +
  \Lambda^{-1} \rb \preceq 0$. Applying Theorem~\ref{thm:mainct}, it
  follows that (\ref{eq:80}) holds with $\eta=0$ for all $0 \le t \le
  t_c$. Hence, $|V^{-1} x(t)| \preceq -\Lambda^{-1} \rb \prec \beta$
  for all $-\bar\tau \le t \le t_c$. Since $x(t)$ is continuous, there
  exists $\alpha > 0$ such that $|V^{-1} x(t)| \preceq \beta$ for all
  $-\bar\tau \le t \le t_c + \alpha$. Consequently, $t_c = \infty$ or
  otherwise the fact that $t_c$ is the largest time instant for which
  (\ref{eq:34}) holds would be contradicted.

  (\ref{item:9}) Since $\Lambda$ is Metzler and Hurwitz, then $-\Lambda$ is an
  M-matrix and $-\Lambda^{-1} \cge 0$ by
  Lemma~\ref{lem:propMetzler}. Since $c\succ 0$, then there exists $j$
  such that $(-\Lambda c)_j > 0$ \cite[Theorem 2.3,
  Ch.6]{berple_book94} and hence $p(c) \neq 0$. Consequently,
  $-\Lambda^{-1}p(c) \neq 0$ and the constraint set of the minimum in
  (\ref{eq:67}) is non-empty. The assumption that $T_0(\beta) \prec
  \beta$ implies that
  \begin{equation}
    \label{eq:69}
    \beta + \Lambda^{-1}\delta(\beta) \succ 0.
  \end{equation}
  Since $-\Lambda^{-1}\cge 0$, $p(c) \cge 0$ and $p(c) \neq 0$, then
  $[-\Lambda^{-1}p(c)]_j > 0$ for every $j$ such that
  $[-\Lambda^{-1}p(c)]_j\neq 0$. These facts jointly with
  (\ref{eq:69}) establish that $\bar\epsilon > 0$.  By (\ref{eq:68}),
  it follows that $\max\{-\Lambda c,0\} = p(c)$ and $p(c\epsilon) =
  p(c)\epsilon$ for every $\epsilon>0$. Hence,
  \begin{equation}
    \label{eq:71}
    0 \cle \max\{-\Lambda c\epsilon,0\} = p(c)\epsilon 
  \end{equation}
  for every $\epsilon > 0$. By (\ref{eq:67}) and (\ref{eq:69}), we
  have
  \begin{equation}
    \label{eq:70}
    -\Lambda^{-1}p(c) \bar\epsilon \cle \beta+\Lambda^{-1}\delta(\beta).
  \end{equation}
  Since $-\Lambda^{-1}\cge 0$, then $[-\Lambda^{-1}p(c)\epsilon]_j <
  [-\Lambda^{-1}p(c)\bar\epsilon]_j$ for every $j$ for which
  $[-\Lambda^{-1}p(c)]_j > 0$ and $0<\epsilon<\bar\epsilon$. Combining
  with (\ref{eq:69})--(\ref{eq:70}), it follows that
  \begin{equation}
    \label{eq:72}
    -\Lambda^{-1}\max\{-\Lambda c\epsilon,0\} \prec \beta+\Lambda^{-1}\delta(\beta),
  \end{equation}
  for every $0<\epsilon<\bar\epsilon$, whence (\ref{item:9}) follows by
  subtracting $\Lambda^{-1}\delta(\beta)$ from each side of the
  inequality (\ref{eq:72}).

  (\ref{item:7}) We first show that, for every $\gamma \in \R^n_+$ and every $k \in
  \integs_+$, there exists a finite time $t_f(k,\gamma)$ such that
    \begin{equation}
      \label{eq:38}
      |V^{-1} x(t)| \preceq T_\gamma^k(\beta),\quad
      \text{for all }t \ge t_f(k,\gamma).
    \end{equation}
    We proceed by induction on $k$. Since $|V^{-1} x(t)| \preceq
    \beta$ for all $t \ge -\bar\tau$, then $\theta(t) \preceq |V|
    \beta$ and hence $|w_i(t)| \preceq \delta_i(|V|\beta)$ for all $t
    \ge 0$. Consider $\bw_i$ as in (\ref{eq:35}) and define
    \begin{equation}
      \label{eq:41}
      \rb \dfn \delta(\beta).
    \end{equation}
    Applying Theorem~\ref{thm:mainct}, it follows that (\ref{eq:29a})
    holds, and hence given $\gamma \in \R^n_+$, there exists
    $t_f(1,\gamma)$ such that $|V^{-1} x(t)| \preceq -\Lambda^{-1} \rb
    + \gamma = T_\gamma(\beta)$ for all $t \ge t_f(1,\gamma)$. The
    claim is thus true for $k=1$. %
    Next, suppose that 
    \eqref{eq:38} is true for some $k \in \integs_+$. It follows from
    (\ref{eq:38}), (\ref{eq:barx}) and (\ref{eq:boundw}) that
    $\theta(t) \preceq |V| T_\gamma^k(\beta)$ and hence $|w_i(t)|
    \preceq \delta_i(|V| T_\gamma^k(\beta))$ for all $t \ge
    t_f(k,\gamma) + \bar\tau$. Define $\bw_i^k \dfn \delta_i(|V|
    T_\gamma^k(\beta))$ and $\rb^k \dfn \delta(
    T_\gamma^k(\beta))$. Taking into account that the system
    is time-invariant, we may apply Theorem~\ref{thm:mainct} to the
    system, considering $t_f(k,\gamma) + \bar\tau$ as the initial
    time. From Theorem~\ref{thm:mainct}, it follows that $\limsup_{t
      \to \infty} |V^{-1} x(t)| \preceq -\Lambda^{-1} \rb^k$. Hence,
    for every $\gamma \in \R^n_+$, there exists $t_f(k+1,\gamma)$ such
    that $|V^{-1} x(t)| \preceq -\Lambda^{-1} \rb^k + \gamma =
    T_\gamma^{k+1}(\beta)$ for all $t \ge t_f(k+1,\gamma)$. Therefore,
    \eqref{eq:38} holds for $k+1$ and the proof by induction is
    complete.

    Next, given $\epsilon \in \R^n_+$, we use
    Lemma~\ref{lem:CNI}(\ref{item:CNIb}) with~$f_\gamma=T_\gamma$
    to obtain $\gamma$ and $k$ so that $T_\gamma^k(\beta) \prec b +
    \epsilon$. For such values of $\gamma$ and $k$,
    we can find, as shown above, a time $t_f(k,\gamma)$ so that
    (\ref{eq:38}) holds. Since this happens for every $\epsilon \in
    \R^n_+$, it follows that $\limsup_{t \to \infty} |V^{-1} x(t)|
    \preceq b$.

\subsection{Proof of Corollary~\ref{cor:nlinv}}
\label{sec:proof-corollary}

  By Theorem~\ref{thm:nlpertct}(\ref{item:5a}), we have $b \preceq
  T_0(\beta) \prec \beta$. Then, for every $\epsilon \in \R^n_{+}$
  small enough, the corresponding $\beta_\epsilon$ satisfies
  $\beta_\epsilon \preceq b + \epsilon \preceq \beta$. Applying $T_0$
  to the inequality $b \preceq \beta_\epsilon \preceq \beta$, and
  iterating, yields $T_0^k(b) = b \preceq T_0^k(\beta_\epsilon)
  \preceq T_0^k(\beta)$ for every $k\in\integs_+$. We thus have
  $\lim_{k \to \infty} T_0^k(\beta_\epsilon) = b \preceq
  T_0(\beta_\epsilon) \prec \beta_\epsilon$. Hence $|V^{-1} x(t)|
  \preceq b \preceq T_0(\beta_\epsilon)$ for all $-\bar\tau \le t \le
  0$. From Theorem~\ref{thm:nlpertct}(\ref{item:6}), then $|V^{-1}
  x(t)| \preceq \beta_\epsilon \preceq b+\epsilon$ for all $t \ge
  -\bar\tau$. Since the latter happens for every $\epsilon \in \R^n_+$
  small enough, then $|V^{-1} x(t)| \preceq b$ for all $t \ge
  -\bar\tau$.

\subsection{Proof of Corollary~\ref{cor:gub}}
\label{sec:proof-corollary-1}

  From (\ref{eq:19}), (\ref{eq:8a}) and (\ref{eq:1a}), it follows that
  \begin{equation*}
    T_0(\beta) \preceq T_\gamma(\beta) \preceq -\Lambda^{-1} 
    [\delta(\beta) + \max\{- \Lambda \gamma, 0\}] \prec \beta.
  \end{equation*}
  Application of Theorem~\ref{thm:nlpertct}(\ref{item:5a}) shows that
  $\lim_{k\to\infty} T_0^k(\beta) = b \succeq 0$,
  Theorem~\ref{thm:nlpertct}(\ref{item:6}) that $|V^{-1} x(t)| \preceq
  \beta$ for all $t \ge -\bar\tau$, and
  Theorem~\ref{thm:nlpertct}(\ref{item:7}) that $\limsup_{t\to \infty}
  |V^{-1} x(t)| \preceq b$.

\subsection{Proof of Lemma~\ref{lem:RLamF}}
\label{sec:proof-lemma-2}

  \ref{item:4}) Since $\Lambda$ is Metzler and Hurwitz and $\bar F\cge 0$, then
  $\Lambda+\bar F$ is Metzler, $-\Lambda^{-1} \cge 0$ and $R\cge 0$.
  Since $\srad(R)<1$ and $R\cge 0$, then $(\id-R)^{-1}
  \cge 0$. We have $\Lambda+\bar F = \Lambda (\id - R)$ and hence
  $(\Lambda+\bar F)^{-1} = (\id - R)^{-1} \Lambda^{-1} \cle 0$. By
  Lemma~\ref{lem:propMetzler}, then $\Lambda+\bar F$ is Hurwitz.

  \ref{item:5}) Since $\Lambda$ is Metzler and $\bar F\cge 0$, then $\Lambda+\bar
  F$ is Metzler. Since $\Lambda+\bar F$ also is Hurwitz, then
  Theorem~\ref{thm:cdlf}\ref{item:2}) establishes that $\Lambda$ is Hurwitz. We
  have $-(\Lambda + \bar F)^{-1} \cge 0$ and $-(\Lambda + \bar F)
  = -\Lambda -\bar F$, where $-\Lambda^{-1} \cge 0$ and $\bar F \cge
  0$. Consequently, $-\Lambda - \bar F$ is a regular splitting of the
  inverse-positive matrix $-(\Lambda+\bar F)$, and hence must be
  convergent, i.e. $\srad(-\Lambda^{-1}\bar F) = \srad(R) < 1$ (see,
  e.g. \cite[Ch.6~\S2]{berple_book94}).

\subsection{Proof of Theorem~\ref{thm:UBlinbnd}}
\label{sec:proof-theorem-2}

  Since $\bar F$ and $\bar w$ have nonnegative entries, then
  $\tilde\delta : \R^n_{+0} \to \R^n_{+0}$ and is CNI. For every
  $\gamma \in \R^n_{+0}$, consider the function $\tilde T_\gamma :
  \R^n_{+0} \to \R^n_{+0}$ defined as
  \begin{equation}
    \label{eq:11}
    \tilde T_\gamma(x) = -\Lambda^{-1} \tilde\delta(x) + \gamma.
  \end{equation}
  Using (\ref{eq:15a}) and \eqref{eq:4a}, we have
  \begin{equation}
    \label{eq:14a}
    \tilde T_\gamma(\beta) = R\beta - \Lambda^{-1} \bar w + \gamma.
  \end{equation}
  Since $\Lambda$ is Metzler and Hurwitz, then $-\Lambda^{-1} \succeq
  0$ by Lemma~\ref{lem:propMetzler}, and hence $R\succeq 0$ and
  $\tilde T_\gamma$ is CNI. Also, by (\ref{eq:5a}), (\ref{eq:14a}) and since
  $\srad(R) < 1$, we have
  \begin{equation}
    \label{eq:32}
    \lim_{k\to\infty} \tilde T_0^k(\beta) = \tilde b = \tilde T_0(\tilde b),
    \quad \text{for all }\beta\in\R^n_{+0}.
  \end{equation}

  (\ref{item:1a}) Applying Lemma~\ref{lem:affine}(\ref{item:LINc}) to
  $\tilde T_0$ (hence $\tilde{b}$ in~\eqref{eq:btilde} has the
  form~\eqref{eq:5a}) we have that the hypotheses of
  Corollary~\ref{cor:nlinv} are satisfied, establishing
  (\ref{item:1a}).

  (\ref{item:2a}) From~\eqref{eq:11} we have $\xi \preceq \tilde
  T_\xi(\beta)$ for every $\xi,\beta \in \R^n_{+0}$. For each $\xi \in
  \R^n_{+0}$, let $v\triangleq -\Lambda^{-1} \max\{-\Lambda \xi,
  0\}$. Note that $v \in \R^n_{+0}$. Applying
  Lemma~\ref{lem:affine}(\ref{item:LINb}) with $v$ as defined and
  $\ell(x)= -\Lambda^{-1} \tilde\delta(x) = R x -\Lambda^{-1} \bar
  w$
  gives $\beta$ satisfying 
  \begin{equation}
    \label{eq:47}
    \ell(\beta)+v=-\Lambda^{-1}
    [\tilde\delta(\beta) + \max\{-\Lambda \xi, 0\}] \prec \beta. 
  \end{equation}
  Then, (\ref{eq:43}) and (\ref{eq:1a}) are satisfied with
  $\gamma=\xi$. Hence, application of Corollary~\ref{cor:gub} and
  recalling (\ref{eq:32}) establishes (\ref{item:2a}).
  
  (\ref{item:8}) Using $\delta(x)$ satisfying~\eqref{eq:52}, define
  $T_\gamma$ as in~(\ref{eq:19}) and consider $\tilde T_\gamma$
  defined in~(\ref{eq:11}). 
  By (\ref{eq:52}) we have $\xi \preceq T_\xi(\beta) \preceq \tilde
  T_\xi(\beta)$ for every $\xi,\beta \in \R^n_{+0}$. For each $\xi \in
  \R^n_{+0}$, we showed above that we can find $\beta$
  satisfying the inequality in~\eqref{eq:47}.
  By (\ref{eq:52}), then
  \begin{equation}
    \label{eq:51}
    -\Lambda^{-1} [\delta(\beta) + \max\{-\Lambda \xi, 0\}] \prec \beta.
  \end{equation}
  Then, (\ref{eq:43}) and (\ref{eq:1a}) are satisfied with
  $\gamma=\xi$. Also, note that (\ref{eq:51}) implies that
  $T_0(\beta)=-\Lambda^{-1} \delta(\beta) \prec \beta$. According to
  Theorem~\ref{thm:nlpertct}(\ref{item:5a}) then
  $b \dfn \lim_{k\to\infty} T_0^k(\beta) \succeq 0$,
  and application of Corollary~\ref{cor:gub} establishes that
  $\limsup_{t\to\infty} |V^{-1} x(t)| \preceq b$. Applying
  Lemma~\ref{lem:affine}(\ref{item:LINd}) with $f(x)=T_0(x)$ and
  $\ell(x)=\tilde T_0(x)$ yields $b = \lim_{k\to\infty} T_0^k(\tilde
  b) \preceq \tilde b$, concluding the proof of (\ref{item:8}).

  (\ref{item:10}) Since $\rho(R)<1$, by Lemma~\ref{lem:RLamF} then
  $\Lambda+\bar F$ is Hurwitz. Application of Theorem~\ref{thm:cdlf}
  with $\bar\Lambda=\Lambda+\bar F$ establishes (\ref{eq:14}).

  (\ref{item:11}) For every $i\in\Ind$, define $p_i(t) \dfn V^{-1}H_i w_i(t)$. Let $x=Vz$
  and rewrite (\ref{eq:system}) as
  \begin{equation}
    \label{eq:15}
    \dot z(t) = \Lambda_{\sw(t)} z(t) + p_{\sw(t)}(t).
  \end{equation}
  Using (\ref{eq:boundw})--(\ref{eq:barx}) with $\bar\tau=0$, it
  follows that, for all $i\in\Ind$,
  \begin{align}
    \label{eq:16}
    |p_i(t)| &\cle \max_{i\in\Ind}\left[\max_{|w_i| \cle
        \delta_i(|Vz(t)|)} |V^{-1} H_i w_i|\right]\\
    \label{eq:22}
    &\cle \psi(|z(t)|) \cle \tilde\delta(|z(t)|) = \bar F |z(t)| + \bar w,
  \end{align}
  where the first inequality in (\ref{eq:22}) follows from $|Vz| \cle
  |V||z|$ and $\delta_i$ CNI. Consider the function $L_z(z) = z^* D z$. We have
  \begin{align*}
    \dot L_z(t,z) &= z^* (\Lambda_{\sw(t)}^* D+D\Lambda_{\sw(t)}) z + 2 \real\{z^*
    D p_{\sw(t)}(t)\}
  \end{align*}
  By Lemma~\ref{lem:propMetzler}e), $z^* (\Lambda_i^* D+D\Lambda_i) z
  \le |z^*| \Met(\Lambda_i^* D+D\Lambda_i) |z|$ for all $z\in\C^n$ and
  all $i\in\Ind$. Arguments identical to those in the proof of
  Theorem~\ref{thm:cdlf}\ref{item:3}) show that $\Met(\Lambda_i^* D+D\Lambda_i)
  \cle M_i' D + D M_i \cle \Lambda' D + D \Lambda$, for all
  $i\in\Ind$. It follows that
  \begin{align}
    \dot L_z(t,z) &\le |z^*| (\Lambda' D+D\Lambda) |z| + 2 |z^*| D
    |p_{\sw(t)}|\notag\\
    \label{eq:18}
    &\le |z^*| [(\Lambda + \bar F)'D+D(\Lambda + \bar F)] |z| + 2|z^*|
    D \bar w,
  \end{align}
  where we have used (\ref{eq:22}).
  Next, taking $x\in\R^n$, we have
  \begin{eqnarray}
    \dot L_z(t,V^{-1}x) &=& x' [A_i' (V^{-1})^* D V^{-1} + (V^{-1})^* D
    V^{-1} A_i] x\notag\\
     & &+ 2 \real \{x' (V^{-1})^* D p_{\sw(t)}(t)\}\notag\\
     &=& x' [A_i' P + P A_i] x + 2 x' P H_{\sw(t)} w_{\sw(t)}\notag\\
     \label{eq:20}
     &=& \dot L(t,x).
  \end{eqnarray}
  Combining (\ref{eq:18})--(\ref{eq:20}) and recalling (\ref{eq:14}),
  it follows that $\dot L(t,x) < 0$ for all $t$ and all $x\in\R^n$
  such that $\|x\|$ is big enough.


\bibliographystyle{plain}
\bibliography{SwitchDTCT}

\end{document}